\documentclass[a4paper,USenglish,modulo]{lipics-v2019} 
\newcommand{\iflong}[1]{#1}
\newcommand{\ifshort}[1]{}

\nolinenumbers

\iflong{\hideLIPIcs}

\usepackage[utf8]{inputenc}
\usepackage{cite}
\usepackage{todonotes}
\usepackage[outercaption]{sidecap} 
\usepackage{wrapfig}
\usepackage{hyperref}
\newcommand\blfootnote[1]{%
  \begingroup
  \renewcommand\thefootnote{}\footnote{#1}%
  \addtocounter{footnote}{-1}%
  \endgroup
}

{
\theoremstyle{theorem}

}

\newtheorem{observation}[theorem]{Observation}

\makeatletter
\g@addto@macro\bfseries{\boldmath}
\makeatother
\usepackage{mdframed}

\graphicspath{{figures/}}

\newtheorem{fact}[theorem]{Fact}

\usepackage{complexity}
\usepackage{thmtools, thm-restate}

\newcommand{\bigoh}{\ensuremath{{\mathcal O}}}
\newcommand{\cH}{\ensuremath{\mathcal H}}
\newcommand{\cG}{\ensuremath{\mathcal G}}

\newcommand{\cA}{\mathcal{A}}
\newcommand{\cB}{\mathcal{B}}
\newcommand{\cC}{\mathcal{C}}
\newcommand{\calP}{\mathcal{P}}

\newcommand{\tw}{\operatorname{tw}}
\newcommand{\aV}{V_{\text{add}}}
\newcommand{\aE}{E_{\text{add}}}
\newcommand{\aEpar}{\aE^{H}}
\newcommand{\aEfree}{\aE^{\neg H}}
\newcommand{\iV}{V_{\text{inc}}}	
\newcommand{\pH}{\cH^{\times}}	
\newcommand{\pG}{\cG^{\times}}

\newcommand{\III}{\mathcal{I}}

\newcommand{\dr}[2][]{\ensuremath{\mathcal{#2}
\ifthenelse { \equal {#1} {} }  
{}   
{_{#1}}   
}
}

\author{Eduard Eiben}{Department of Computer Science, Royal Holloway, University of London, Egham, United Kingdom}{eduard.eiben@rhul.ac.uk}{https://orcid.org/0000-0003-2628-3435}{}
\author{Robert Ganian}{Algorithms and Complexity Group, TU Wien, Vienna, Austria}{rganian@ac.tuwien.ac.at}{https://orcid.org/0000-0002-7762-8045}{Robert Ganian acknowledges support by the Austrian Science Fund (FWF, project P31336).}
\author{Thekla Hamm}{Algorithms and Complexity Group, TU Wien, Vienna, Austria}{thamm@ac.tuwien.ac.at}{}{Thekla Hamm acknowledges support by the Austrian Science Fund (FWF, projects P31336 and W1255-N23).}
\author{Fabian Klute}{Algorithms and Complexity Group, TU Wien, Vienna, Austria}%
{fklute@ac.tuwien.ac.at}{https://orcid.org/0000-0002-7791-3604}{}
\author{Martin Nöllenburg}{Algorithms and Complexity Group, TU Wien, Vienna, Austria}{noellenburg@ac.tuwien.ac.at}{https://orcid.org/0000-0003-0454-3937}{}

\authorrunning{E.~Eiben, R.~Ganian, T.~Hamm, F.~Klute, M.~Nöllenburg}

\Copyright{Eduard Eiben, Robert Ganian, Thekla Hamm, Fabian Klute, Martin Nöllenburg}

\title{Extending Partial 1-Planar Drawings}
\titlerunning{Extending Partial 1-Planar Drawings}

\begin{document}
\maketitle
\begin{abstract}
Algorithmic extension problems of partial graph representations such as planar graph drawings or geometric intersection representations are of growing interest in topological graph theory and graph drawing. 
In such an extension problem, we are given a tuple $(G,H,\cH)$ consisting of a graph $G$, a connected subgraph $H$ of $G$ and a drawing $\cH$ of $H$, and the task is to extend $\cH$ into a drawing of $G$ while maintaining some desired property of the drawing, such as planarity.

In this paper we study the problem of extending partial 1-planar drawings, which are drawings in the plane that allow each edge to have at most one crossing. 
In addition we consider the subclass of IC-planar drawings, which are 1-planar drawings with independent crossings.
Recognizing 1-planar graphs as well as IC-planar graphs is \NP-complete and the \NP-completeness easily carries over to the extension problem. Therefore, our focus lies on establishing the tractability of such extension problems in a weaker sense than polynomial-time tractability. Here, we show that both problems are fixed-parameter tractable when parameterized by the number of edges missing from $H$, i.e., the edge deletion distance between $H$ and $G$. 
The second part of the paper then turns to a more powerful parameterization which is based on measuring the vertex+edge deletion distance between the partial and complete drawing, i.e., the minimum number of vertices and edges that need to be deleted to obtain $H$ from $G$.
	\iflong{\blfootnote{\noindent \emph{A shortened version of this article has been accepted for presentation and publication at the 47th International Colloquium on Automata, Languages and Programming (ICALP 2020).}}}
\end{abstract}

\newpage
\newcommand{\edg}{embedding graph}
	\section{Introduction}
	In the last decade, algorithmic extension problems of partial \emph{planar} graph drawings have received a lot of attention in the fields of graph algorithms and graph theory as well as in graph drawing and computational geometry. 
	In this problem setting, the input consists of a planar graph $G$, a connected subgraph $H$ of $G$, and a planar drawing $\cH$ of $H$; 
	the question is then whether $\mathcal H$ can be extended to a planar drawing of $G$.
	This extension problem is motivated from applications in network visualization, where important patterns (subgraphs) are required to have a special layout, or where new vertices and edges in a dynamic graph must be inserted into an existing (partial) connected drawing, which must remain stable to preserve its mental map~\cite{mels-lam-95}.	
\iflong{
	A major result on the extension of partial planar drawings is the linear-time algorithm of Angelini et al.~\cite{adfjkp-tppeg-15} which can answer the above question as well as provide the desired planar drawing of $G$ (if it exists)---showing that constrained inputs do not increase the complexity of planarity testing. 	Their result is complemented by a combinatorial characterization via forbidden substructures~\cite{jkr-ktppeg-13}.}%

\iflong{
		The result of Angelini et al.\ is in contrast to other algorithmic extension problems, e.g., on graph coloring of perfect graphs~\cite{DBLP:journals/jgt/KratochvilS97} or 3-edge coloring of cubic bipartite graphs~\cite{f-cepepbg-03}, which are both polynomially tractable but become \NP-complete if partial colorings are specified.
	Again more related to extending partial planar drawings, it is well known by Fáry's Theorem that every planar graph admits a planar straight-line drawing, but testing straight-line extensibility of partial planar straight-line drawings is generally \NP-hard~\cite{p-epsd-06}.
	Polynomial-time algorithms are known for certain special cases, e.g.,
	if the subgraph $H$ is a cycle drawn as a convex polygon, and the straight-line extension must be inside~\cite{cegl-dgpwpofpa-12} or outside~\cite{mnr-ecpdg-15} the polygon.
	Yet, if only the partial drawing is a straight-line drawing and the added edges can be drawn as polylines, Chan et al.~\cite{cfglms-dpespg-15} showed that if a planar extension exists, then there is also one, 
	where all new edges are polylines with at most a linear number of bends. 
	This generalizes a classic result by Pach and Wenger~\cite{DBLP:journals/gc/PachW01} that any $n$-vertex planar graph can be drawn on any set of $n$ points in the plane using polyline edges with $O(n)$ bends.
	Similarly, level-planarity testing takes linear time~\cite{jlm-lptlt-98}, but testing the extensibility of partial level-planar drawings is \NP-complete~\cite{br-pclp-17}.
	Recently, Da Lozzo et al.~\cite{ddf-eupgd-19} studied the extension of partial upward planar drawings for directed graphs, which is  generally \NP-complete, but some special cases admit polynomial-time algorithms.
	Other related work also studied extensibility problems of partial representations for specific graph classes \cite{kkorss-eprpuig-17,kkosv-eprig-17,kkos-eprscg-15,kkkw-eprfgpg-12,cfk-eprcg-13,cdkms-crpgeprh-14,cggkl-pvrep-18}.
	
	In this paper, we study the algorithmic extension problem of partial drawings of 1-planar graphs, one of the most natural and most studied generalizations of planarity~\cite{klm-ab1-17,dlm-sgdbp-19}.
	A 1-planar graph is a graph that admits a drawing in the plane with at most one crossing per edge.
	The definition of 1-planarity dates back to Ringel (1965)~\cite{r-sk-65}  and since then the class of 1-planar graphs has been of considerable interest in graph theory, graph drawing and (geometric) graph algorithms, see the recent annotated bibliography on 1-planarity by Kobourov et al.~\cite{klm-ab1-17} collecting 143 references. 
	More generally speaking, interest in various classes of beyond-planar graphs (not limited to, but including 1-planar graphs) has steadily been on the rise~\cite{dlm-sgdbp-19} in the last decade.


	Unlike planarity testing, recognizing 1-planar graphs is \NP-complete~\cite{gb-agewce-07,km-mo1h1t-13}, even if the graph is a planar graph plus a single edge~\cite{cm-aepgmcn1h-13}. It is known, however, that 1-planarity testing is fixed-parameter tractable (\FPT) for the vertex-cover number, the cyclomatic number, and the tree-depth of the input graph $G$, but it remains \NP-complete for graphs of bounded bandwidth, pathwidth, or treewidth~\cite{bce-pc1-18}.
	Moreover, restrictions of 1-planarity have been studied, such as independent-crossing (IC) planarity, which additionally requires 
	that no two crossed edges are incident~\cite{albertson08,LiottaM16,Brandenburg18a,YangW019}. 
	The recognition problem for IC-planar graphs remains \NP-complete~\cite{bdeklm-rdig-16}. 

	On the other end of the planarity spectrum, Arroyo et al.~\cite{adp-esd-19,Arroyo2019-ext1edge} studied drawing extension problems, where the number of crossings per edge is not restricted, yet the drawing must be \emph{simple}, i.e., any pair of edges can intersect in at most one point. 
	They showed that the simple drawing extension problem is \NP-complete~\cite{adp-esd-19}, even if just one edge is to be added~\cite{Arroyo2019-ext1edge}.
}

\ifshort{\subparagraph{Contributions.}}
\iflong{\subsection*{Contributions}}
	Given a graph $G$, a connected subgraph $H$, and a 1-planar drawing $\cH$ of $H$, the \textsc{1-Planar Extension} problem asks whether $\mathcal H$ can be extended by inserting the remaining vertices $\aV = V(G) \setminus V(H)$ and edges $\aE = E(G) \setminus E(H)$ of $G$ into $\mathcal H$ while maintaining the property of being 1-planar. The \textsc{IC-Planar Drawing Extension} problem is then defined analogously, but for IC-planarity.
	
	The \NP-completeness of these extension problems is a simple consequence of the \NP-completeness of the recognition problem~\cite{gb-agewce-07,km-mo1h1t-13,bdeklm-rdig-16} (see also Section~\ref{sec:hardness}).	
	With this in mind, the aim of this paper is to establish the tractability of the problems when $\cH$ is almost a complete 1-planar drawing of $G$. 
	To capture this setting, we turn to the notion of \emph{fixed-parameter tractability}~\cite{DowneyF13,CyganFKLMPPS15} and consider two natural parameters which capture how complete $\cH$ is:
	\ifshort{the \emph{edge deletion distance} between $H$ and $G$ (denoted by $k$), and the \emph{vertex+edge deletion distance} between $H$ and $G$ (denoted by $\kappa$).	
	}\iflong{\begin{itemize}
	\item $k$ is the \emph{edge deletion distance} between $H$ and $G$, and
	\item $\kappa$ is the \emph{vertex+edge deletion distance} between $H$ and $G$.
	\end{itemize}
	}%
	More precisely, $k$ is equal to $|E(G)\setminus E(H)|$ and
	$\kappa$ is equal to $ |V(G)\setminus V(H)| + |E(G[V(H)])\setminus E(H)| $.
	We refer to Section~\ref{sec:hardness} for formal definitions and a discussion of the parameters.
	
	\medskip
	
	\noindent After introducing necessary notation in Section~\ref{sec:prelims} and 
	introducing the problem formally in Section~\ref{sec:hardness},
	 we consider the edge deletion distance $k$ in \textbf{Section~\ref{sec:1pfpt}}. 
	Our first result is:
	
	\begin{restatable*}{mainthm}{fptkcourcelle}
		\label{thm:1Pfpt}
		\textsc{1-Planar Drawing Extension} is \FPT\ when parameterized by $k$.
	\end{restatable*}
	

	The proof of Theorem~\ref{thm:1Pfpt} involves the use of several ingredients: 
	\begin{enumerate}
	\item Introducing and developing a notion of \emph{patterns}, which are combinatorial objects that capture critical information about the potential interaction of newly added edges with $\cH$;
	\item a pruning procedure that reduces our instance to an equivalent sub-instance where $H$ has treewidth bounded in $k$;
	\item an \edg, which carries information about the drawing $\cH$; and finally
	\item completing the proof by constructing a formula $\Phi$ in Monadic Second Order Logic to check whether a pattern can ``fit'' in the \edg, using Courcelle's Theorem~\cite{Courcelle90}.
\end{enumerate}

	Next, we turn towards the question of whether one can obtain an \emph{efficient} fixed-parameter algorithm for the extension problem. In particular, due to the use of Courcelle's Theorem~\cite{Courcelle90} to model-check $\Phi$, the algorithm obtained in the proof of Theorem~\ref{thm:1Pfpt} will have a prohibitive dependency on the parameter $k$. In this direction, we note that it is not immediately obvious how one can design an efficient and 
``formally clean'' purely combinatorial algorithm for the pattern-fitting task (i.e., the task we relegate to model checking $\Phi$ in the \edg). At the very least, using a direct translation of the model-checking procedure would come at a significant cost in terms of presentation clarity.
	
	That being said, one can observe that the main reason for the use of patterns is that it is not at all obvious where (i.e., in which cell of the drawing) one should place the vertices used to extend $\cH$. 
	Indeed, our second result for parameter $k$ assumes that $\aV = \emptyset$ and avoids using Courcelle's Theorem. 
	
	\begin{restatable*}{mainthm}{fptkdirect}
		\label{thm:1PfptE}
		\textsc{1-Planar Drawing Extension} parameterized by \(k\) can be solved in time \(\bigoh(k^{2k}\cdot n^{\bigoh(1)})\) if $V(G)=V(H)$.
	\end{restatable*}

	This algorithm uses entirely different techniques---notably, it prunes the search space for inserting each individual edge via a combination of geometric and combinatorial arguments, and then applies exhaustive branching. We note that the techniques used to prove Theorem~\ref{thm:1Pfpt} and~\ref{thm:1PfptE} can be directly translated to also obtain analogous results for the IC-planarity setting.
	
	\medskip
	
In \textbf{Section~\ref{sec:kappa}}, we turn our attention to the vertex+edge deletion distance $\kappa$ as a parameter, which represents a more relaxed way of measuring how complete $\cH$ is than $k$---indeed, while $\kappa\leq k$, it is easy to construct instances where $\kappa=1$ but $k$ can be arbitrarily large.
	For our third result, we start with IC-planar drawings.
	
	\begin{restatable*}{mainthm}{fptkappaic}
		\label{thm:ICPfpt}
	\textsc{IC-Planar Drawing Extension} is \FPT\ parameterized by $\kappa$.
	\end{restatable*}
	
	The proof of Theorem~\ref{thm:ICPfpt} requires a significant extension of the toolkit developed for Theorem~\ref{thm:1Pfpt}. The main additional complication lies in the fact that the number of edges that are missing from $H$ is no longer bounded by the parameter. To deal with this, we show that the added vertices can only connect to the boundary of a cell in a bounded number of ``ways'' (formalized via a notion we call \emph{regions}), and we use this fact to develop a more general notion of patterns and embedding graphs than those used for Theorem~\ref{thm:1Pfpt}.
	
	
	
	\medskip
	
	\noindent Finally, in \textbf{Section~\ref{sec:2vtcs}}, we present a first step towards the tractability of \textsc{1-Planar Drawing Extension} parameterized by $\kappa$. We note that the techniques developed for the other parameterizations and problem variants cannot be applied to solve this case---the main difference compared to the setting of Theorem~\ref{thm:ICPfpt} is that the ``missing'' vertices can be incident to many edges with crossings, which prevents the use of our bounded-size patterns to capture the behavior of new edges.
 As our final contribution, we investigate the special case of $\kappa=2$, i.e., when adding two new vertices.
	
%
	\begin{restatable*}{mainthm}{xpkappaonep}
		\label{thm:1Pxp}
	\textsc{1-Planar Drawing Extension} is polynomial-time tractable if $\kappa\leq 2$.
	\end{restatable*}
	
We note that even this, seemingly very restricted, subcase of \textsc{1-Planar Drawing Extension} was non-trivial and required the combination of several algorithmic techniques (this contrasts to the case of $|\aV| = 1$, whose polynomial-time tractability is a simple corollary of one of our lemmas). In particular, the algorithm uses a new two-step ``delimit-and-sweep'' approach: first, we apply branching to find a curve with specific properties that bounds the instance by a left and right ``delimiter''. The second step is then a left-to-right sweep of the instance that iteratively pushes the left delimiter towards the right one while performing dynamic programming combined with branching and network-flow subroutines.

Albeit being a special case, we believe these delimited instances with two added vertices can play a role in a potential XP algorithm parameterized by \(\kappa\)---the existence of which we leave open for future work.

\smallskip
\noindent \textbf{Further Related Work.}\quad
		In addition to the given related work on extension problems, it is also worth noting that identifying a substructure of bounded treewidth and applying Courcelle's Theorem to decide an MSO formula on it has been preciously used for a graph drawing problem by Grohe~\cite{grohe04}, namely to identify graph drawings of bounded crossing number.
		Both the way in which one arrives at bounded treewidth and the nature of the employed MSO formula are substantially different from our approach, which is not surprising as the problem of generating drawings from scratch and the problem of extending partial drawings are in general fundamentally different.
		Specifically in the case of generating drawings, the MSO formula could essentially encode the existence of a drawing with bounded crossing number by inductively planarizing crossings of pairs of edges; here the planarity of the planarization can of course be captured via excluded \(K_{3,3}\) and \(K_5\) minors by MSO.
		This approach is not possible in our setting.
		There are examples of 1-planar graphs which have partial drawings which cannot be extended to a 1-plane drawing. Thus a planarization with respect to the added parts of a solution needs to be compatible with the partial drawing and cannot be encoded by an MSO formula straightforwardly.

	\section{Preliminaries}
	\label{sec:prelims}
\iflong{\noindent \textbf{Graphs and Drawings in the Plane.}	}
	Let $G$ be a simple graph, $ V(G) $ its vertices, and $ E(G) $ its edges.
	We use standard graph terminology~\cite{Diestel}. For $r \in \mathbb{N}$, we write $[r]$ as shorthand for the set $\{1, \ldots, r\}$. \iflong{The length of a path is the number of edges contained in that path.}		
	\ifshort{We also assume a basic understanding of parameterized complexity theory~\cite{DowneyF13,CyganFKLMPPS15}, Monadic Second Order (MSO) Logic and Courcelle's Theorem~\cite{Courcelle90}.}
		
	A \emph{drawing} $\cG$ of $G$ in the plane $\mathbb R^2$ is a function that maps each vertex $v \in V(G)$ to a distinct point $\cG(v) \in \mathbb R^2$ and each edge $e=uv \in E(G)$ to a simple open curve $\cG(e) \subset \mathbb R^2$ with endpoints $\cG(u)$ and $\cG(v)$.
	In a slight abuse of notation we often identify a vertex $v$ and its drawing $\cG(v)$ as well as an edge $e$ and its drawing $\cG(e)$. \ifshort{Throughout the paper we will assume that: (i) no edge passes through a vertex other than its endpoints, (ii) any two edges intersect in at most one point, which is either a common endpoint or a proper \emph{crossing} (i.e., edges cannot touch), and (iii) no three edges cross in a single point.}\iflong{
	We say that a drawing $\cG$ is a \emph{good drawing} (also known as  a \emph{simple topological graph}) if (i) no edge passes through a vertex other than its endpoints, (ii) any two edges intersect in at most one point, which is either a common endpoint or a proper \emph{crossing} (i.e., edges cannot touch), and (iii) no three edges cross in a single point.
	For the rest of this paper we assume that every drawing $\cG$ is good.}
	For a drawing $\cG$ of $G$ and $e\in E(G)$, we use $\cG-e$ to denote the drawing of $G-e$ obtained by removing the drawing of $e$ from $\cG$, and for $J\subseteq E(G)$ we define $\cG-J$ analogously.
	
\iflong{	We say that $\cG$ is \emph{planar} if no two edges $e_1, e_2 \in E(G)$ cross in $\cG$; if the graph $G$ admits a planar drawing, we say that $G$ is planar.
	A planar drawing $\cG$ subdivides the plane into connected regions called \emph{faces}, where exactly one face, the \emph{outer} (or \emph{external}) face is unbounded. 
	}\ifshort{We assume that readers are familiar with the notion of \emph{planarity} and \emph{faces}.} The \emph{boundary} of a face is the set of edges and vertices whose drawing delimits the face.
	Further, $\cG$ induces for each vertex $v \in V(G)$ a cyclic order of its neighbors by using the clockwise order of its incident edges.
	This set of cyclic orders is called a \emph{rotation scheme}.
	Two planar drawings $\cG_1$ and $\cG_2$ of the same graph $G$ are \emph{equivalent} if they have the same rotation scheme and the same outer face; equivalence classes of planar drawings are also called \emph{embeddings}.
	A \emph{plane} graph is a planar graph with a fixed embedding. 

	A drawing $\cG$ is \emph{1-planar} if each edge has at most one crossing and a graph $G$ is \emph{1-planar} if it admits a 1-planar drawing.
Similarly to planar drawings, 1-planar drawings\iflong{ also define a rotation scheme and} subdivide the plane into connected regions, which we call \emph{cells} in order to distinguish them from the faces of a planar drawing.
	The \emph{planarization} $G^\times$ of a 1-planar drawing $\cG$ of $G$ is a graph $G^\times$ with $V(G) \subseteq V(G^\times)$ that introduces for each crossing $c$ of $\cG$ a \emph{dummy vertex} $c \in V(G^\times)$ and that replaces each pair of crossing edges $uv, wx$ in $E(G)$ by the four \emph{half-edges} $uc, vc, wc, xc$ in $E(G^\times)$, where $c$ is the crossing of $uv$ and $wx$. In addition all crossing-free edges of $E(G)$ belong to $E(G^\times)$.
		Obviously, $G^\times$ is planar and the drawing $\pG$ of $G^\times$ corresponds to $\cG$ with the crossings replaced by the dummy vertices.

\iflong{		
			With $ H+e $ we denote the graph $ H $ with the additional edge $ e \in E(G) \setminus E(H) $ added to it. Further, for a 1-planar drawing $ \cH $ of $ H $ we denote with $ \cH + \gamma(e) $ the 1-planar drawing that we get by fixing a specific curve $ \gamma(e) $ for $ e $ and adding it to the drawing $ \cH $. We say edge $ e $ is drawn into $ \cH $ with $ \gamma(e) $. If $ \gamma(e) $ is clear we omit it and only write $ \cH + e $.
	
\medskip 

\noindent \textbf{Monadic Second Order Logic.} 
	We consider \emph{Monadic Second Order} (MSO) logic on (edge-)labeled
	directed graphs in
	terms of their incidence structure, whose universe contains vertices and
	edges; the incidence between vertices and edges is represented by a
	binary relation. We assume an infinite supply of \emph{individual
		variables} $x,x_1,x_2,\dots$ and of \emph{set variables}
	$X,X_1,X_2,\dots$. The \emph{atomic formulas} are 
	$V x$ (``$x$ is a vertex''), $E y$ (``$y$ is an edge''), $I xy$ (``vertex $x$
	is incident with edge $y$''), $x=y$ (equality),
	$P_a x$ (``vertex or edge $x$ has label $a$''), and $X x$ (``vertex or
	edge $x$ is an element of set $X$'').  \emph{MSO formulas} are built up
	from atomic formulas using the usual Boolean connectives
	$(\lnot,\land,\lor,\rightarrow,\leftrightarrow)$, quantification over
	individual variables ($\forall x$, $\exists x$), and quantification over
	set variables ($\forall X$, $\exists X$).

	\emph{Free and bound variables} of a formula are defined in the usual way. 
	To indicate that the set of free individual variables of formula $\Phi$
	is $\{x_1, \dots, x_\ell\}$ and the set of free set variables of formula $\Phi$
	is $\{X_1, \dots, X_q\}$ we write $\Phi(x_1,\ldots, x_\ell, X_1,
	\dots, X_q)$. If $G$ is a graph, $v_1,\ldots, v_\ell\in V(G)\cup E(G)$ and $S_1, \dots, S_q
	\subseteq V(G)\cup E(G)$ we write $G \models \Phi(v_1,\ldots, v_\ell, S_1, \dots, S_q)$ to denote that
	$\Phi$ holds in $G$ if the variables $x_i$ are interpreted by the vertices or edges $v_i$, for $i\in [\ell]$, and the variables $X_i$ are interpreted by the sets
	$S_i$, for $i \in [q]$. 
	
	
	The following result (the well-known Courcelle's Theorem~\cite{Courcelle90}) 
	shows that if $G$ has bounded treewidth~\cite{RobertsonS84} then we
	can find an assignment $\varphi$ to the set of free variables $\mathcal{F}$ with $G \models \Phi(\varphi(\mathcal{F}))$ (if one exists) in linear time. 
	
	\begin{fact}[Courcelle's Theorem~\cite{Courcelle90,ArnborgLS91}]\label{fact:MSO} 
		Let $\Phi(x_1,\dots,x_\ell, X_1,\dots, X_q)$ be a fixed MSO formula with free individual variables $x_1,\dots,x_\ell$ and free set variables $X_1,\dots,X_\ell$, and let $w$ a
		constant. Then there is a linear-time algorithm that, given a labeled
		directed graph $G$ of treewidth at most $w$, 
		either outputs  $v_1,\ldots, v_\ell\in V(G)\cup E(G)$ and $S_1, \dots, S_q	\subseteq V(G)\cup E(G)$ such that $G \models \Phi(v_1,\ldots, v_\ell, S_1, \dots, S_q)$ or correctly identifies that no such vertices $v_1,\ldots, v_\ell$ and sets $S_1, \dots, S_q$ exist.
	\end{fact}
	}
	
	\section{Extending 1-Planar Drawings}\label{sec:hardness}
		Given a graph $ G $ and a subgraph $ H $ of $ G $ with a 1-planar drawing $ \cH $ of $ H $, we say that a drawing $ \cG $ of $ G $ is an \emph{extension} of $ \cH $ if the planarization $H^\times$ of $\cH$ and the planarization $\pG$ of $\cG$ restricted to $\cH^\times$ have the same embedding.
	 We formalize our problem of interest as:
	
	\begin{mdframed}\label{prob:extension}
		\textsc{1-Planar Drawing Extension}\\
		{\itshape Instance:}  A graph \(G\), a connected subgraph \(H\) of \(G\), and a 1-planar drawing \(\cH\) of \(H\).\\
		{\itshape Task:} Find an 1-planar extension of \(\cH\) to \(G\), or correctly identify that there is none.
	\end{mdframed}

	The \textsc{IC-Planar Drawing Extension} problem is then defined analogously. Both problem definitions follow previously considered drawing extension problems, where the connectivity of \(H\) is considered a well-motivated and standard assumption~\cite{mels-lam-95,mnr-ecpdg-15,HongN08}.
	

	Given an instance $(G,H,\cH)$ of \textsc{1-Planar Drawing Extension}, a \emph{solution} is a 1-planar drawing $\cG$ of $G$ that is an extension of $\cH$.
	We refer to \(\aV := V(G) \setminus V(H)\) as the \emph{added vertices} and to \(\aE := E(G) \setminus E(H)\) as the \emph{added edges}.  Let $\iV=\{ v\in V(H) \mid \exists vw\in \aE\}$, i.e., $\iV$ is the set of vertices of $H$ that are incident to at least one added edge.
	We also distinguish added edges whose endpoints are already part of the drawing, and added edges with at least one endpoint yet to be added into the drawing---notably, we let
\ifshort{$\aEpar := \left\lbrace vw \in \aE \mid v, w \in V(H) \right\rbrace
	\text{ and } \aEfree := \aE \setminus \aEpar.$}
\iflong{
	\[\aEpar := \left\lbrace vw \in \aE \mid v, w \in V(H) \right\rbrace
	\text{ and } \aEfree := \aE \setminus \aEpar.\]
	}
	This distinction will become important later, since it opens up two options for how to quantify how ``complete'' the drawing of $\cH$ is.
It is worth noting that, without loss of generality, we may assume each vertex in $\aV$ to be incident to at least one edge in $\aE$ and hence $|\aV\cup \iV|\leq 2|\aE|$. \iflong{Furthermore, it will be useful to assume that $\cH$, $\pH$, $\cG$ and $\pG$ are all drawn atop of each other in the plane, i.e., vertices and edges are drawn in the same coordinates in $\cH$ and $\cG$---this allows us to make statements such as ``a solution $\cG$ draws vertex $v\in \aV$ inside face $f$ of $\pH$''.}
		
Given the \NP-completeness of recognizing 1-planar~\cite{gb-agewce-07,km-mo1h1t-13} and IC-planar~\cite{bdeklm-rdig-16} graphs we get as an immediate consequence that also the corresponding extension problems are \NP-complete. 

\iflong{	\begin{proposition}
	 \label{thm:np-hard-2}
	 
		\textsc{1-Planar Drawing Extension} and \textsc{IC-Planar Drawing Extension} are \NP-complete even if all added edges have at least one endpoint that can be placed freely, i.e., if \(\aEpar = \emptyset\).
	\end{proposition}
	\begin{proof}
		The claim is an immediate consequence of the \NP-completeness of recognizing 1-planar~\cite{gb-agewce-07,km-mo1h1t-13} and IC-planar graphs~\cite{bdeklm-rdig-16}.
		For the reduction, let the subgraph $H$ consist of a single, arbitrary vertex $v \in V(G)$, which we draw  in $\cH$ at some fixed position in the plane.
		The position of $v$ is no restriction to the existence of a 1-planar or IC-planar drawing, since any such drawing can be translated such that $v$ is mapped to the selected position. So a 1-planar/IC-planar extension of $\cH$ exists if and only if $G$ is 1-planar/IC-planar.
	\end{proof}
}

In view of the \NP-completeness  of the problem, it is natural to ask about its complexity when $H$ is nearly ``complete'', i.e., we only need to extend the drawing $\cH$ by a small part of $G$. In this sense, deletion distance represents the most immediate way of quantifying how far $H$ is from $G$, and the parameterized complexity paradigm~\cite{DowneyF13,CyganFKLMPPS15} offers complexity classes that provide a more refined view on ``tractability'' in this setting.

The most immediate way of capturing the completeness of $H$ in this way is to parameterize the problem via the \emph{edge deletion distance} to $G$---formalized by setting $k=|\aE|$. The aim of Section~\ref{sec:1pfpt} is to establish the fixed-parameter tractability of \textsc{1-Planar Drawing Extension} parameterized by~$k$. A second parameter that we consider is the \emph{vertex+edge deletion distance} to $G$, i.e., the minimum number of vertices and edges that need to be deleted from $G$ to obtain $H$. We call this parameter $\kappa$ and set $\kappa=|\aV|+|\aEpar|$. 
The parameterization by \(\kappa\) is the topic of Section~\ref{sec:kappa} and~\ref{sec:2vtcs}. 
Since we can always assume that each added vertex is incident to at least one added edge, $|\aV|+|\aEpar|\leq |\aE|$ and so parameterizing by $\kappa$ leads to a more general (and difficult) parameterized problem.

\section{Using Edge Deletion Distance for Drawing Extensions}
\label{sec:1pfpt}
The main goal of this section is to establish the fixed-parameter tractability of \textsc{1-Planar Drawing Extension} parameterized by the edge deletion distance $k$. 

We note that one major obstacle faced by a fixed-parameter algorithm is that it is not at all obvious how to decide where the vertices in $\aV$ should be drawn in an augmented drawing of $H$. As a follow-up, we will show that when $\aV=\emptyset$ (i.e., $V(H)=V(G)$), it is possible to obtain a more self-contained combinatorial algorithm with a significantly better runtime; this is presented in Subsection~\ref{sub:justedges}.

\subsection{A Fixed-Parameter Algorithm for \textsc{1-Planar Drawing Extension}}
\label{subsec:1pfpt}

Our first step towards a proof of the desired tractability result is the definition of a \emph{pattern}, which is a combinatorial object capturing essential information about a potential 1-planar extension of~$\cH$. 
%
The formal definition of pattern is given in Definition~\ref{def:pat}.
Definition~\ref{def:derived} then defines the notion of \emph{derived patterns},
which create a link between solutions to an instance of 1-\textsc{Planar Drawing Extension} and patterns.

To given an intuition of the patterns, assume that a pattern consists of a tuple $(S,Q,C)$ and
let $ (G, H, \cH) $ be a 1-\textsc{Planar Drawing Extension} instance. 
Then, the general intuition is that $S$ represents the set of faces in $ \pH $ which contain at least a part of the drawing of an edge in $\aE$ in a hypothetical 1-planar extension $ \cG $ of $\cH$. Crucially, our aim is to keep the size of patterns bounded in $k$, and so we only ``anchor'' $S$ to $\pH$ by storing information about which faces will contain individual edges in $\aE$, vertices from $\aV$, and be adjacent to individual vertices in $\iV$; this is captured by the mapping $Q$. The third piece of information we store is $C$, which represents the cyclic order of how edges in $\aE$ exit or enter the boundary of each face (including the case where an edge crosses through an edge into the same face, i.e., occurs twice when traversing the boundary of that face).

\newcommand{\cros}{\textnormal{crossing}}

\iflong{
\begin{definition}
\label{def:pat}
A \emph{pattern} for an instance $(G,H,\cH)$ is a tuple $(S,Q,C)$ where 
\begin{enumerate}
\item $S$ is a set of at most $2k$ elements;
\item $Q$ is a mapping from $\aV\cup \aE \cup \iV$ which maps:
\begin{itemize}
\item vertices in $\aV$ to elements of $S$;
\item edges in $\aE$ to ordered pairs of elements of $S$;
\item vertices in $\iV$ to subsets of $S$.
\end{itemize}
\item $C$ is a mapping from $S$ that maps each $s\in S$ to a cyclically ordered multiset of pairs $((e_1,q_1),(e_2,q_2),$ $\dots, (e_\ell,q_\ell))$, where each $e_i$ is in $\aE$ and each $q_i$ is in $\iV\cup\{$$\cros$$\}$. Here $\cros$ is a special new symbol signifying a crossing point.
Moreover, $C$ must satisfy the following conditions: 
\begin{itemize}
\item for each $s\in S$ and each tuple $(e,q)\in C(s)$ such that $q\in \iV$, it must hold that $s\in Q(q)$ and $e$ is incident to $q$ in $G$;
\item for each $e\in \aE$ and $s\in S$, if $e$ occurs in at least one tuple in $C(s)$, then $s\in Q(e)$ and $C(s)$ contains at most two tuples of the form $(e,*)$, where $*$ is an arbitrary element;
\item for each $s\in S$, each tuple occurs at most once in $C(s)$ with the exception of tuples containing ``$\cros$'', which may occur twice.
\end{itemize}
\end{enumerate}
\end{definition}
}
\ifshort{
\begin{definition}
\label{def:pat}
A \emph{pattern} for an instance $(G,H,\cH)$ is a tuple $(S,Q,C)$ where 
\begin{enumerate}
\item $S$ is a set of at most $2k$ elements;
\item $Q$ is a mapping from $\aV\cup \aE \cup \iV$ which maps:
\emph{\textbf{(a)}} vertices in $\aV$ to elements of $S$, \emph{\textbf{(b)}} edges in $\aE$ to ordered pairs of elements of $S$, and \emph{\textbf{(c)}} vertices in $\iV$ to subsets of $S$.
\item $C$ is a mapping from $S$ that maps each $s\in S$ to a cyclically ordered multiset of pairs $((e_1,q_1),(e_2,q_2),$ $\dots, (e_\ell,q_\ell))$, where each $e_i$ is in $\aE$ and each $q_i$ is in $\iV\cup\{\cros\}$.
Moreover, $C$ must satisfy the following conditions: \emph{\textbf{(a)}} for each $s\in S$ and each tuple $(e,q)\in C(s)$ such that $q\in \iV$, it must hold that $s\in Q(q)$ and $e$ is incident to $q$ in $G$; \emph{\textbf{(b)}} for each $e\in \aE$ and $s\in S$, if $e$ occurs in at least one tuple in $C(s)$, then $s\in Q(e)$ and $C(s)$ contains at most two tuples of the form $(e,*)$, where $*$ is an arbitrary element; \emph{\textbf{(c)}} for each $s\in S$, each tuple occurs at most once in $C(s)$ with the exception of tuples containing ``$\cros$'', which may occur twice.
\end{enumerate}
\end{definition}
}

\newcommand{\numpat}{\#\text{pat}}
Let $\mathcal{P}$ be the set of all patterns for our considered instance $(G,H,\cH)$. Let $\numpat(k)=2k\cdot (2^{2k})^{3k} \cdot ((2k)!\cdot 2^{3k})^{2k}$ and note that $|\mathcal{P}|\leq \numpat(k)\in 2^{\bigoh(k^2\log k)}$.
	In particular, the number of possible patterns can be bounded by first considering $2k$ options for $|S|$, multiplying this by the at most $(2^{2k})^{3k}$-many ways of choosing $Q$, and finally multiply this by the number of choices for $C$ which can be bounded as follows: for each $s\in S$, $C(s)$ is a set that forms a subset (of size at most $2k$) of the $3k$-cardinality set of tuples (note that $e=\{a,b\}\in \aE$ can only occur in the tuples $(e,a)$, $(e,b)$ and $(e,\cros)$).

The intuition behind patterns will be formalized in the next definition, which creates a link between solutions to our instance and patterns.


\begin{definition}
\label{def:derived}
Let $ (G, H, \cH) $ be a 1-\textsc{Planar Drawing Extension} instance. For each solution $\cG$ of $ (G, H, \cH) $ we define a derived pattern $P = (S,Q,C)$ as follows:
\begin{itemize}
\item $S$ is the set of faces of $\pH$ which have a non-empty intersection with $ \cG(e) $ for some $ e \in \aE $.
\item For $v\in \aV$ we set $ Q(v) $ to the face \(f\) of \(\pH\) for which $ \cG(v) $ lies inside $ f $,
for $e\in \aE$ we set $ Q(e) $ to the set of at most two faces which have a non-empty intersection with $ \cG(e) $,
and for $w\in \iV$ we set $ Q(w) $ to all faces in $S$ incident to $w$ in $ \cG $.
\item For a face $ s \in S $ we consider all edges $ e = uv \in \aE $ with a non-empty intersection between $ \cG(e) $ and $s$. 
It follows that there is an edge $ e' \in E(H) $ on the boundary of $s$ such that $ \cG(e) $ crosses $ \cG(e') $, or $ u \in \iV $ and $ u $ is on the boundary of $ s $, or both. 
We set $ C(s) $ as the ordered set of these crossing points or vertices when traversing $ s $ in clockwise~fashion.

\end{itemize} 
\end{definition}

Our next task is to define \emph{valid patterns}; generally speaking, these are patterns which are not malformed and could serve as derived patterns for a hypothetical solution. 
	One notable property that every valid pattern must satisfy is that all vertices and edges mapped by $Q$ to some $s\in S$ can be drawn in a 1-planar way while respecting $C(s)$. 

\begin{definition}\label{def:valid}
For an instance $(G,H,\cH)$, a pattern $P=(S,Q,C)$ is \emph{valid} if there exists a \emph{pattern graph} $G_P$ with a 
$1$-planar drawing $\cG_P$ satisfying the following properties:
\begin{itemize}
\item $\aV\cup \iV\subseteq V(G_P)$ and $\aE\subseteq E(G_P)$.
\item $\cG_P-\aE$ is a planar drawing.
\item $S$ is a subset of the inner faces of $\cG_P-\aE$.
\item Each $v\in \aV$ is contained in the face $Q(v)$ of $\cG_P-\aV-\aE$.
\item Each $e\in \aE$ is contained in the face(s) $Q(e)$ of $\cG_P-\aE$.
\item Each $v\in \iV$ is incident to the faces $Q(v)$ of $\cG_P-\aE$.
\item When traversing the inner side of the boundary of each face $s$ of $\cG_P-\aE$ in clockwise fashion, the order in which each edge $e\in \aE$ is seen in $\cG_P$ together with the information whether $e$ crosses here or ends in its endpoint in $\iV$, is precisely $C(s)$.\end{itemize}
\end{definition}

Note that the instance $(G,H,\cH)$ in the definition of a valid pattern is only important to define $\aV$, $\iV$, and $\aE$. Moreover,
observe that for each solution $\cG$ of an instance $(G,H,\cH)$, the derived pattern is valid by definition. An illustration of a pattern graph is provided in Part (a) of Figure~\ref{fig:patternandembeddinggraph}.
We also remark that, when comparing a pattern graph to a hypothetical solution which draws an edge into the outer face of $\cH$, we will map the outer face to an inner face of the pattern graph.

\begin{figure}
	\begin{subfigure}[t]{0.48\textwidth}
		\centering
		 \vspace{-0.4cm}
		\includegraphics[scale=0.9]{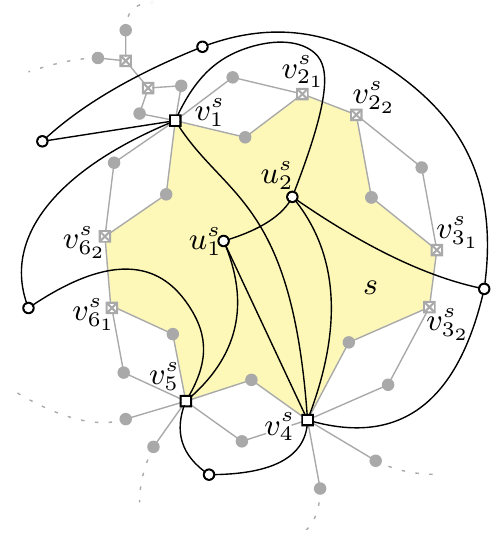} \vspace{-0.4cm}
		\caption{A pattern graph as constructed in Lemma~\ref{lem:validpattern}. The face representing $ s \in S $ is yellow, gray disks are dummy vertices. Black circles are in $ \aV $. Squares are either in $ \iV $ or represent crossings. \vspace{-0.4cm}}
		\label{fig:pattern}		
	\end{subfigure}
	\hfill
	\begin{subfigure}[t]{0.48\textwidth}
		\centering		
		 \vspace{-0.4cm}
		\includegraphics[scale=0.9]{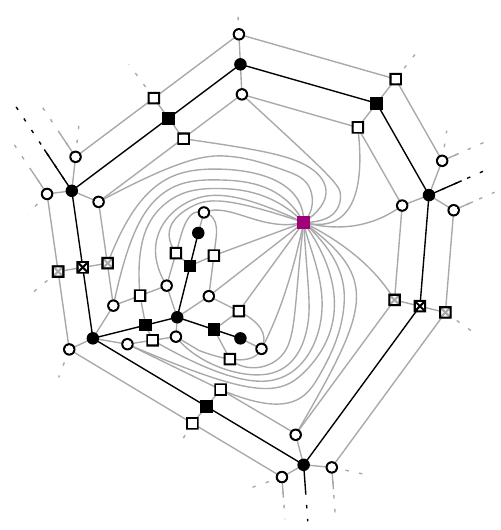}
		\caption{An example of an embedding graph. The white vertices are shadow-vertices, the purple one the face vertex, and gray edges got added.}
		\label{fig:embeddinggraph}
	\end{subfigure}
	\caption{Examples for the definition of a pattern and the \edg. \vspace{-0.4cm}}
	\label{fig:patternandembeddinggraph}
\end{figure}

\begin{lemma}
\label{lem:validpattern}
	Given pattern $P = (S,Q,C)$,
	in time $\bigoh((k!)^k\cdot k^{2k + 1})$ we can either construct a pattern graph $G_P$ together with the drawing $\cG_P$ satisfying all the properties of Definition~\ref{def:valid} or decide that $P$ is not valid.
\end{lemma}
\iflong{
\begin{proof}
	Our first step is to construct a possible planarization of the pattern graph $ G_P $, as follows. 
	For each $s\in S$ we add a cycle $ \mathcal C(s) $ with $ |C(s)| $ vertices.
	Let $ v_i^s $ be the $ i $-th vertex on the cycle $ \mathcal C(s) $, then we mark $ v_i^s $ as the vertex representing the second element of the tuple $ c_i \in C(s) $.
	In case $ |C(s)| < 3 $ we add one or two more dummy vertices to $ \mathcal C(s) $, hence every $ \mathcal C(s) $ is at least a triangle.
	For each $ s \in S $ we then subdivide every edge in $ \mathcal C(s) $ by a dummy vertex.
	Next, for each $ u \in \aV $ we add  a vertex $ u_s $ to the graph and also all corresponding edges to the individual vertices in the cycles constructed from $C(s)$, i.e., $u_s$ will be adjacent to $ v_i^s $ if and only if the first element of the tuple $c_i\in C(s)$ is an edge incident to $u$. 
	We then identify two vertices $ v_1,v_2 $ if they represent the same vertex in $ \iV $ or a crossing point of the same edge in $ \aE $.
	
	Our last step towards the desired planarization is to pre-assign the crossings between the drawings of two added edges. Since the number of such crossings is upper-bounded by $k$, we can branch over which pairs of edges cross in time at most $k^{2k}$. 
	Let $ (u,v),(u',v') $ be two edges that cross, then replace them by a new vertex $ x $ with edges $ (u,x)$,$(v,x)$,$(u',x) $, and $ (v',x) $. We call the resulting graph (for one particular branch) $G_P'$.
	
	It remains to determine if there is a plane drawing $ \cG_P' $ of $ G_P' $ conforming with the properties of Definition~\ref{def:valid}.
	If we find such a drawing, we are left to replace the vertices representing crossings by actual edges representing the edges in $ \aE $.
	A crossing introduced in the last step, we replace it simply by two edges that cross.
	For a vertex $ v_i^s $ in a cycle $ \mathcal C(s) $ for $ s \in S $ we handle by splitting $ v_i^s $ into two vertices $ v_{i_1}^s $ and $ v_{i_2}^{s} $ and adding the edge $ v_{i_1}^sv_{i_2}^s $.
	This edge can then be crossed by the edge replacing the two half edges incident to $ v_i^s $. 
	Let $ \cG_P $ be the resulting 1-plane drawing.
	This drawing and its represented graph $ G_P $ fulfill Definition~\ref{def:valid} and hence $ P $ is valid.
	In case no such drawing can be found we return that $ P $ is not valid.

	To compute a possible drawing we first observe that the graph $ G_P' $ has only $ O(k) $ vertices.
	This is easy to see after realizing that we added at most $ 2|\aE| $ vertices for crossings and at most $ |\iV| + |\aE| $ dummy vertices.
	Further every vertex in $ G_P' $ has degree in $ O(k) $.
	For every crossing vertex the degree is in fact six or four, and for every dummy vertex two. For a vertex $ u_s $ representing $ u \in \aV $ the degree is at most $ |\aE| $. 
	Finally for every vertex $ v_i^s $ representing some $ v \in \iV $ we can upper bound the degree by $ 5k $ since $ v_i^s $ is incident to at most $ |\aE| + 2|S| $ edges.
	In total $ G_P' $ has $ O(k) $ vertices and max-degree $ O(k) $, which enables us to iterate all possible rotation schemes in time $ O((k!)^k) $.
	If a rotation scheme implies a plane drawing $ \cG_P' $ we can further check in time $ O(k) $ the conditions of Definition~\ref{def:valid}.
		
	Assume the above construction fails even though there exists a pattern graph $ G_P $ with $ 1 $-planar drawing $ \cG_P $ and properties as in Definition~\ref{def:valid}; to obtain a contradiction, we will show that in fact we can construct a graph $ G_P' $ and drawing $ \cG_P' $ as above.
	By definition we find for each $ s \in S $ a cycle $ \mathcal C(s) $ in $ G_P $ such that  $ V(\mathcal C(s)) \cap \aV = \emptyset $ and $ C(s) \subseteq V(\mathcal C(s)) $, and the vertices and crossings in $ C(s) $ appear in that clockwise order on $ \mathcal C(s) $.
	Further let $ \gamma $ be the simple, closed curve described by the curves representing the edges $ E(\mathcal C(s)) $ in $ \cG_P $.
	Again by definition $ \mathcal C(s) $ exists such that every $ \cG_P(V) $ for $ v \in \aV $ with $ Q(v) = s $ lies in the interior of $ \gamma $ and every $ \cG_P(e)$ for $ e \in \aE $ with $ s \in Q(e) $ has a non-empty intersection with the interior of $ \gamma $.
	To build $ G_P' $ from $ G_P $ consider the planarization $ G_P^\times $ derived from $ \cG_P $. Again we find the cycles as above, now with every crossing being a vertex and some edges in $ \aE $ being represented by two half-edges.
	First, delete for every cycle $ \mathcal C(s) $ all vertices $ v \in V(G_P) $ that do not represent a vertex in $ \aV $ and lie inside the by $ \mathcal C(s) $ described curve.
	Secondly replace every vertex $ v \in V(\mathcal C(s)) $, which is neither representing a crossing nor is in $ \iV $, by an edge.
	Finally add the dummy vertices to every $ \mathcal C(s) $.
	The resulting graph is exactly the graph the above algorithm would have computed, a contradiction.
\end{proof}
}

Next, we will define an annotated (``labeled'') graph representation of $\cH$ and \(\pH\)'s faces.
\newcommand{\eH}{H^*}
The \emph{\edg} $\eH$ of $\cH$ is obtained from $\pH$ by:
\begin{enumerate}
\item\label{edg:Step1} subdividing each uncrossed edge $e$ (resulting in vertex $v_e$);
\item creating a vertex for each face \(f\) of \(\pH\) (resulting in vertex $v_f$);
\item traversing the boundary of each face $f$\footnote{formally, we draw a curve in \(f\) that closely follows
	the boundary 
	until it forms a closed curve.} and whenever we see a vertex $v$ (including the vertices created in Step~\ref{edg:Step1}) we create a shadow copy of $v$ and place it right next to $v$ in the direction we saw $v$ from. Add a cycle connecting the shadow vertices we created in $f$ in the order they were created, and direct it in clockwise fashion\footnote{Note that this may create multiple shadow copies of a vertex. The reason we use shadow copies of vertices instead of using the original vertices is that when traversing the inner boundary of a face, a vertex may be seen multiple times, and such shadow-vertices allow us to pinpoint from which part of the face we are visiting the given vertex.};
\item connecting $v_f$ to all shadow-vertices created by traversing $f$, and all shadow copies of a vertex $v$ to the original~$v$.
\end{enumerate}

Observe that the \edg\ is a connected plane graph. We label the vertices of the \edg\ to distinguish original vertices, edge-vertices, face-vertices, crossing-vertices and shadow-vertices, and use at most $2k$ special labels to identify vertices in $\iV$. An illustration of the \edg\ is provided in Part (b) of Figure~\ref{fig:patternandembeddinggraph}.
Next, we show that it suffices to restrict our attention to the parts of \(\eH\) which are ``close'' to vertices in $\iV$. 
\iflong{For a drawing $\cG$ of a graph $G$ and a subgraph $G'$ of $G$, let the \emph{restriction} of $\cG$ to $G'$ be the drawing obtained by removing $\cG(v)$ for each $v\in V(G)\setminus V(G')$ and $\cG(e)$ for each $e\in E(G)\setminus E(G')$.}

\begin{lemma}
\label{lem:notfar}
Let $I=(G,H,\cH)$ be an instance of \textsc{1-Planar Drawing Extension}. Let $Z$ be the set of all vertices in $\eH$ of distance at least $4k+7$ from each vertex in $\iV$. Let $G'$, $H'$, and $\cH'$ be obtained by deleting all vertices in $Z$ from $G$, $H$, and $\cH$ respectively.~%
Then:
\begin{enumerate}
\item If $I$ is a \textsc{YES}-instance, then each connected component of $G'$ contains at most one connected component of $H'$\footnote{This can be seen not to hold in general if we allow \(H\) to be disconnected.};
\item $I$ is a \textsc{YES}-instance if and only if for each connected component $A$ of $H'$ the restriction of $\cH'$ to $H'[A]$ can be extended to a drawing of the connected component of $G'$ containing~$A$. Moreover, given such a 1-planar extension for every connected component of $G'$, we can output a solution for $I$ in linear time.
\end{enumerate}
\end{lemma}

\iflong{
\noindent We split the proof of Lemma~\ref{lem:notfar} into proofs for the two individual points.

\begin{proof}[Proof of Point 1]
%
%
%
%
	For the sake of contradiction let $J$ be a connected component of $G'$ that contains two distinct connected components $H_1'$ and $H_2'$ of $H'$. Since $J$ is a connected component, there must be a path $P$ from a vertex $v_1\in H_1'$ to a vertex $v_2\in H_2'$ in $J-(H_1'\cup H_2')$, and moreover $P$ must have length at most $k$. By definition, both $v_1$ and $v_2$ are in $\iV$. To complete the proof, it suffices to show that in any solution $\cG$, $v_1$ and $v_2$ have distance at most $4k+4$ in $\eH$.
	
	
	Moreover, in any solution $\cG$, two consecutive vertices of $P$ are either drawn in the same face of $\pH$ or in two adjacent faces of $\pH$. Observe that the distance in $\eH$ between two face-vertices for the faces that share an edge is $4$, and that the distance from an original vertex $v$ to a face-vertex of a face incident to $v$ is $2$. Therefore, if $(G,H,\cH)$ is a \textsc{YES}-instance, then the distance between $v_1$ and $v_2$ in $\eH$ must be at most $4k+4$.
	\end{proof}
}

\iflong{
\begin{proof}[Proof of Point 2]
The forward direction is obvious. For the backward direction, let $G_1,\ldots, G_r$ be the connected components of $G'$ and for $i\in [r]$ let $H_i$ and $\cH_i$  be the restriction of $H'$ and $\cH'$, respectively, to $G_i$. Moreover, let $\pH_i$ be the planarization derived from $\cH_i$ and note that $H_i$ is connected for all $i\in [r]$ by Point 1. 
		Now let us fix an arbitrary $i\in [r]$ such that $H_i$ is not empty and let $\cG_i$ be a 1-planar extension of $\cH_i$ to $G_i$. 
		
		Observe that each face of $\pH$ is completely contained in precisely one face of $\pH_i$.
Moreover, if a face $f$ of $\pH_i$ contains at least two faces $f_1$ and $f_2$ of $\pH$, then both $v_{f_1}$ and $v_{f_2}$ are at distance at least $4k+4$ of any vertex in $\iV\cap V(H_i)$ in $\eH$. Indeed, if this were not the case, then w.l.o.g.\ the vertices on the boundary of $v_{f_1}$ would have distance at most $4k+6$ from some $w\in \iV\cap V(H_i)$ in $\eH$, which would mean that $f_1$ is also a face in $\pH_i$. By the same distance-counting argument introduced at the end of the Proof of Point 1, This implies that no edge in a path $P$ of $G$ from a vertex $v\in H_i$ whose internal vertices all lie in $\aV$ can be drawn in any face of $\pH$ contained in $f$.

To complete the proof, let $G_1, \ldots, G_p$, $p\le r$ be the connected components of $G'$ that contain a vertex in $H$ and $G_{p+1},\ldots, G_r$ the remaining connected components of $G'$. We obtain a solution $\cG$ to the instance $I$ by simply taking the union of $\cH$ and $\cG_i$ for $i\in [p]$ and then for $i\in \{p+1,\ldots, r \}$ shifting $\cG_i	 $ so that $\cG_i$ do not intersect any other part of the drawing. 
\end{proof}
}

\newcommand{\diam}{\text{diambound}}
Since $|\iV|\leq 2k$, Lemma~\ref{lem:notfar} allows us to restrict our attention to a subgraph of diameter at most $(4k+7)\cdot 2 \cdot 2k=16k^2+28k$. This will be especially useful in view of the following known fact, that allows us
to assume that the treewidth of our instances is bounded.

\begin{proposition}[\hspace{1sp}\cite{RobertsonS84}]\label{prop:radius_treewidth}
	A planar graph $ G $ with radius at most $r$ has treewidth at most $3r+1$.
\end{proposition} 

	\begin{lemma}
\label{lem:gettw}
\textsc{1-Planar Drawing Extension} is \FPT\ parameterized by $k+\tw(\eH)$ if and only if it is \FPT\ parameterized by $k$,
where $\eH$ is the embedding graph of $\cH$. 
\end{lemma}

\begin{proof}
The backward direction is trivial. For the forward direction, assume that that there exists an algorithm $\cB$ which solves \textsc{1-Planar Drawing Extension} in time $f(k+\tw(\eH))\cdot |V(G)|^c$ for some constant $c$ and computable function $f$. Now, consider the following algorithm $\cA$ for \textsc{1-Planar Drawing Extension}: $\cA$ takes an instance $(G_0,H_0,\cH_0)$ and constructs $(G_1,H_1,\cH_1)$ by applying Lemma~\ref{lem:notfar}. Recall that by Point 1 of Lemma~\ref{lem:notfar}, $(G_0,H_0,\cH_0)$ is either \textsc{NO}-instance, in which case $\cA$ correctly outputs ``NO'', or each connected component of $G_1$ contains at most one connected component of $H_1$. 

Now let us consider a connected component $\cC$ of $G_1$ and the \edg\ $\eH_1[\cC]$ of $\cH_1[\cC]$ and let $v_f$ be a face-vertex in $\eH_1[\cC]$. If $v_f$ is at distance at least $4k+9$ from every vertex in $\iV\cap\cC$ in $\eH_1[\cC]$, then every vertex on the boundary of $f$ is at distance at least $4k+7$ from every vertex $w\in \iV\cap\cC$ in $\eH_1[\cC]$. Let $v$ be an arbitrary vertex incident to $f$ in $\cH_1[\cC]$. Since each face of $\pH_0$ is completely contained in precisely one face of $\pH_1[\cC]$, it follows that $v$ is at distance at least $4k+7$ from each vertex $w\in\iV\cap\cC$ in $\eH$. Because $v\in V(H_1[\cC])$, this contradicts the fact that every vertex in $V(H_1)$ is at distance at most $4k+6$ from a vertex $w\in\iV$ in $\eH$. Hence, every face-vertex in $\eH_1[\cC]$ is at distance at most $4k+8$ from a vertex in $\iV\cap\cC$. Moreover, every vertex in $\eH_1[\cC]$ is at distance at most $2$ from some face-vertex and there are at most $2k$ vertices in $\iV\cap\cC$. Therefore, the radius, and by Proposition~\ref{prop:radius_treewidth} the treewidth, of $\eH_1[\cC]$ is bounded by $\bigoh(k^2)$.


Now, for each connected component $\cC$ of $G_1$, we solve the instance $(G_1[\cC],H_1[\cC],\cH_1[\cC])$ using algorithm $\cB$. If $\cB$ determines that at least one such (sub)-instance is a \textsc{NO}-instance, then $\cA$ correctly outputs ``NO''. Otherwise, $\cA$ outputs a solution for $(G_0,H_0,\cH_0)$ that it computes by invoking the algorithm given by Point 2 of Lemma~\ref{lem:notfar}. 
To conclude, we observe that $\cA$ is a fixed-parameter algorithm parameterized by $k$ and its correctness follows from Lemma~\ref{lem:notfar}.
\end{proof}


We now have all the ingredients we need to establish our tractability result.

\fptkcourcelle


\newcommand{\edge}{\operatorname{adj}}
\newcommand{\conn}{\operatorname{conn}}
\newcommand{\arc}{\operatorname{arc}}
	\newcommand{\freeVar}{\mathcal{F}}

\ifshort{
\begin{proof}[Proof Sketch]
We prove the theorem by showing that \textsc{1-Planar Drawing Extension} is fixed-parameter tractable parameterized by $k+\tw(\eH)$, which suffices thanks to Lemma~\ref{lem:gettw}.

To this end, consider the following algorithm $\cA$. Initially, $\cA$ loops over all of the at most $\numpat(k)$ many patterns, tests whether each pattern is valid or not using Lemma~\ref{lem:validpattern}, and stores all valid patterns in a set $\calP$. Next, it branches over all valid patterns in $\calP$, and for each such pattern $P=(S=\{s_1,\dots,s_\ell\},Q,C)$ it constructs an MSO formula $\Phi_P(\freeVar)$, where $\freeVar$ is a set of at most $7k$ free variables specified later, 
 the purpose of which is to find a suitable ``placement'' for $P$ in $\cH$ by finding an interpretation in the \edg\ $\eH$. In particular, $\Phi_P$ uses the free variables in $\freeVar$ to find a suitable face-vertex $x_i$ for each $s_i\in S$ and a suitable crossing point for each edge mapped to two elements of $S$, while also guaranteeing that the cyclic orders specified by $C$ are adhered to. Once we find a suitable placement for $P$ in $\cH$, the algorithm constructs an extension by topologically ``inserting'' the pattern graph $G_P$ into the identified faces of $\pH$ and using the crossing points as well as vertices in $\iV$ as ``anchors''.
\end{proof}
}

\iflong{
\begin{proof}
	
	
	

We prove the theorem by showing that \textsc{1-Planar Drawing Extension} is fixed-parameter tractable parameterized by $k+\tw(\eH)$, which suffices thanks to Lemma~\ref{lem:gettw}.

To this end, consider the following algorithm $\cA$. Initially, $\cA$ loops over all of the at most $\numpat(k)$ many patterns, tests whether each pattern is valid or not using Lemma~\ref{lem:validpattern}, and stores all valid patterns in a set $\calP$. Next, it branches over all valid patterns in $\calP$, and for each such pattern $P=(S=\{s_1,\dots,s_\ell\},Q,C)$ it constructs an MSO formula $\Phi_P(\freeVar)$, where $\freeVar$ is a set of at most $7k$ free variables specified later, 
 the purpose of which is to find a suitable ``embedding'' for $P$ in $\cH$ by finding an interpretation in the \edg\ $\eH$.

In the following we will formally define the MSO formula $\Phi_P(\freeVar)$. Recall that the vertices of $\eH$ have the following labels: a label $v$ for every vertex $v\in \iV$ and then the labels $O$, $E$ ,$F$, $S$ which represent original, edge-, face-, crossing-, and shadow-vertices, respectively. 
For vertices $x, y$, let $\edge(x,y)$ be a formula stating that $x$ and $y$ are adjacent vertices, and $\conn(x, y, X)$ a formula stating that there is a directed path\footnote{Recall that edges between shadow vertices are directed.} from $x$ to $y$ with all inner vertices in $X$. 

The set of free variables $\freeVar$ of $\Phi_P(\freeVar)$ consists of:
\begin{itemize}
	\item $x_1,\ldots, x_\ell$, where $x_i$ corresponds to a single element $s_i$ in $S$; 
	\item $y_1, \ldots, y_{k'}$, where $y_i$ corresponds to an edge $e_i\in \aE$ that crosses an edge in $H$---formally, $(e_i,\cros)\in C(s_j)$ for some $j\in [\ell]$ (Note that this $e_i$ could either cross from one face of $\pH$ to another, but also could cross an edge of $H$ that is incident to a single face in $\pH$);
	\item for each $i\in [\ell]$, we have $z_1^i, \ldots, z_{|C(s_i)|}^i$ -- where $z_j^i$ correspond to $j$-th element of $C(s_i)$ (after fixing some arbitrary first element in the cyclic ordering). 
\end{itemize}

Note that $\ell\leq 2k$, $k'\leq k$, and the total number of variables of the form $z_j^i$ is upper-bounded by $4k$ since each edge $e\in \aE$ can occur in at most $4$ tuples across all cyclic orders in a valid pattern (in particular, $e$ may start in some $v_1\in \iV$, cross to a second face, and then end in some $v_2\in \iV$).
The formula $\Phi_P(\freeVar)$ is then the conjunction of the following subformulas:

\begin{enumerate}
		\item checkFaces$(\freeVar)$, which ensures that $x_i$'s are assigned to distinct face-vertices and is the conjunction of:
		\begin{itemize}
			\item $P_Fx_i$, for all $i\in[\ell]$ and 
			\item $x_i\neq x_j$ for all $1\le i < j\le \ell$;
		\end{itemize}
		\item checkEdges$(\freeVar)$, which ensures that $y_i$'s are assigned to distinct edge-vertices and is the conjunction of:
		\begin{itemize}
			\item $P_Ey_i$, for all $i\in[k']$ and 
			\item $y_i\neq y_j$ for all $1\le i < j\le k'$;
		\end{itemize}
		\item checkShadow$(\freeVar)$, which ensures that $z_j^i$'s are assigned to shadow-vertices that are adjacent to $x_i$:
		\begin{itemize}
			\item for all $i\in [\ell]$ and $j\in [|C(s_i)|]$ we have $P_Sz_j^i\wedge \edge(x_i,z_j^i)$
		\end{itemize}
		\item checkCrossings$(\freeVar)$, which ensures that the edge-vertex $y_p$,
		corresponding to an edge $e_p\in \aE$ crossing an edge in $H$ incident to faces $s_1, s_2$, is adjacent to $z_{j_1}^{i_1}$ and $z_{j_2}^{i_2}$ corresponding to the two pairs $(e_p,\cros)$ in $C(s_1)$ and $C(s_2)$, respectively: 
		\begin{itemize}
			\item for all $p\in [k']$ and the corresponding $z_{j_1}^{i_1}$ and $z_{j_2}^{i_2}$, checkCrossings$(\freeVar)$ contains $\edge(z_{j_1}^{i_1}, y_p)\wedge \edge(y_p, z_{j_2}^{i_2})$.
		\end{itemize}
		\item check$\iV(\freeVar)$, which ensures that if incidence between an edge $e\in \aE$ and a vertex $v\in \iV$ is realized in the face $s_i$ (i.e., $(e,v)\in C(s_i)$), then the $z_j^i$ corresponding to $(e,v)$ in $C(s_i)$ is adjacent to $v$.
		\begin{itemize}
			\item For all $i\in [\ell]$ and all $(e,v)\in \aE\times\iV$ such that $(e,v)$ corresponds to $z_j^i$, check$\iV(\freeVar)$ contains $\exists u \left( P_vu\wedge \edge(x_i,z_{j}^i) \wedge \edge(z_{j}^i, u)\right)$.
		\end{itemize}
		\item checkCyclicOrder$(\freeVar)$, which ensures that $z_j^i$'s occur in the cyclic order around the face-vertex $x_i$ given by $C(s_i)$:
		\begin{itemize}
			\item for all $i\in [\ell]$ and $j\in [|C(s_i)|]$ checkCyclicOrder$(\freeVar)$ contains 
			
			$\exists X \left( (\forall x (Xx \rightarrow (\edge(x_i,x)\wedge (x\neq z_1^i) \wedge\ldots\wedge (x\neq z_{|C(s)|}^i) ))\wedge \conn(z_j^i, z_{j+1}^i, X)) \right)$, where \\$z_{|C(s_i)|+1}^i = z_{1}^i$.
		\end{itemize}
\end{enumerate}

Clearly, the length of the formula $\Phi_P(\freeVar)$ is bounded by a function of $k$. 
Hence, we can use Fact~\ref{fact:MSO} to, in time $f(k,\tw(\eH))\cdot |\eH|$ for some computable function $f$, either decide that $\eH\not \models \Phi_P(\freeVar)$ or find an assignment  $\phi: \freeVar\rightarrow V(\eH)$ such that $\eH \models \Phi_P(\phi(\freeVar))$. 

Given the assignment $\phi$, we construct an extension $\dr{G}$ of $\cH$ as follows. Let $G_P$ and $\cG_P$ be a pattern graph and its 1-planar drawing of $G_P$, respectively, satisfying the properties of Definition~\ref{def:valid}. We can construct $G_P$ and $\cG_P$ in time bounded by a function of $k$ by Lemma~\ref{lem:validpattern}.
By Definition~\ref{def:valid}, $\cG_P-\aV-\aE$ is a planar drawing where $S$ is a set of non-outer faces of $\cG_P-\aV-\aE$. The subformula checkFaces$(\freeVar)$ ensures that $\phi$ maps $x_1,\dots, x_\ell$ in a way which captures a bijection between faces of $\cG_P-\aV-\aE$ and face-vertices of $\eH$ (which represent faces of $\cH$). Furthermore, given a face $f_i$, $i\in [\ell]$, of $\cG_P-\aV-\aE$ which corresponds in this way to $s_i$ and $x_i$, when traversing the inside of the boundary of $f_i$ in a clockwise fashion the pairs $(e,t)\in \aE\times (\iV\cup \{\cros\})$ are seen precisely in the same order as in $C(s_i)$. This, thanks to subformula checkCyclicOrder$(\freeVar)$, is in turn the same cyclic order as the order of vertices $z_1^i,\ldots,z_{|C(s_i)|}^i$ in the neighborhood of $x_i$ in $\eH$.

We will now glue the interior of $f_i$ inside the face of $H$ represented by $x_i$ such that we glue the elements of $C(s_i)$ precisely on the corresponding vertices in $\{z_1^i,\ldots,z_{|C(s_i)|}^i\}$, with one small exception: if $t$ consecutive elements of $C(s_i)$ are mapped to the same shadow vertex $u$ adjacent to $x_i$, we create $t$ copies of $u$ (by subdividing edges between $u$ and a neighboring shadow vertex) and perform the gluing to these copies in a way which preserves the cyclic ordering. To extend the edges to the vertices in $\iV$ and to connect the edges in two different faces, we concatenate them with the $z_j^i$-$v$ and $z_{i'}^i$-$z_{j'}^j$ paths guaranteed by check$\iV(\freeVar)$ and checkCrossings$(\freeVar)$, respectively. 

Since all added edges are drawn in $\cG_P$, it follows that all added edges will be drawn once this procedure ends. Furthermore, inside the faces of $\cH$, the added edges cross precisely in the same way as they crossed in $\cG_P$ and if an added edge crosses between two faces of $\cH$, then first it also crosses an edge of $\cG_P-\aV-\aE$, second it crosses a edge of $\cH$ that is not crossed yet, and third at most one added edge crosses this edge. The second and third point of the previous sentence are guaranteed by subformula checkEdges$(\freeVar)$. Therefore, each edge crosses at most one other edge and what we get is indeed a 1-planar drawing of $G$ that extends $\cH$. 

On the other hand, given a solution $\cG$ to $(G,H, \cH)$, it is straightforward to verify that if $\phi$ assigns $x_i$'s to the face-vertices for faces that intersect at least a part of an added edge in $\cG$, $y_i$'s to the edge-vertices for the edges of $\cH$ that are crossed by an added edge, and $z_j^i$'s to the shadow-vertices corresponding to the intersections of added edges with the boundary of the face corresponding to $x_i$, then $\eH\models \Phi_P(\phi(\freeVar))$ for the derived pattern $P$ for $\cG$. 
\end{proof}
}

%
%

\subsection{A More Efficient Algorithm for Extending by Edges Only}
%
\label{sub:justedges}

In this subsection we obtain a more explicit and efficient algorithm than in Theorem~\ref{thm:1Pfpt} for the case where $V(G)=V(H)$. 
%
The idea underlying the algorithm is to iteratively identify sufficiently many 1-planar drawings of each added edge into \(\cH\) that can either all be extended to a 1-planar drawing of \(G\), or none of them can, which allows us to 
branch over a small number of possible drawings for that edge.

Let $X=\bigcup_{uv\in \aE} \{u,v\}$ be the set of all endpoints of edges in $\aE$, and let us fix an order of the added edges by enumerating \(\aE = \{e_1, \dotsc, e_k\}\).
Now, consider a 1-planar drawing \(\cH_i\) of \(H_i := H + \{e_1, \dotsc, e_{i-1}\}\) and assume that we want to add \(e_i\) as a curve \(\gamma(e_i)\).
For a cell $f$ in $\cH_i + \gamma(e_i)$ and vertices $x_1,x_2$ on the boundary of $f$, we denote by $b_{\gamma(e_i)}(f,x_1,x_2)\subseteq E(H_{i+1})$ the edges on the $x_1$-$x_2$-path along the boundary of $f$ which traverses this boundary in counterclockwise direction. We explicitly note that $b_{\gamma(e_i)}(f,x_1,x_2)$ does not contain any half-edges.
In this way \(b_{\gamma(e_i)}(f,x_1,x_2)\) is the set of edges of \(H_{i+1}\) on the \(x_1\)-\(x_2\)-path along the boundary of \(f\) that are not crossed in \(\cH_i + \gamma(e_i)\), and hence may still be crossed by drawings of \(e_{i + 1}, \dotsc, e_k\) in a 1-planar extension of \(\cH_i + \gamma(e_i)\) to \(G\).

Let $ \gamma_1(e_i) $ and $ \gamma_2(e_i) $ be two possible curves for $ e_i $ to be drawn into $\cH_i$. 
Then we call $\gamma_1(e_i)$ and $ \gamma_2(e_i) $ \emph{\(\{e_{i+1}, \dotsc, e_k\}\)-partition equivalent} 
if there is a bijection $ \pi $ from the cells of $ \cH_i + \gamma_1(e_i) $  to the cells of $ \cH_i + \gamma_2(e_i) $ such that
\begin{itemize}
	\item the vertices in \(X\) on the boundaries of the cells
	are invariant under \(\pi\), i.e., for each cell $f$ whose boundary intersects $X$ precisely in $X'$ it must hold that $\pi(f)$ intersects $X$ precisely in $X'$ as well; and
	\item for each pair of cells \(f, f'\) of \(\cH_i + \gamma_1(e_i)\) and ordered pairs of (not necessarily pairwise distinct) vertices \((x_1, x_2), (x_1', x_2')\in X^2\) that
	\begin{itemize}
	\item are on the boundary of \(f\) and \(f'\), respectively, and
	\item the counterclockwise \(x_1\)-\(x_2\)-path and the counterclockwise \(x_1'\)-\(x_2'\)-path along the boundaries of $f$ and $f'$, respectively, does not contain any inner vertices in \(X\),
	\end{itemize}
	the following must hold:
	\[
	\begin{cases}
	\mbox{if }
	\begin{aligned}[t]
	& \vert b_{\gamma_1(e_i)}(f,x_1,x_2) \cap b_{\gamma_1(e_i)}(f',x_1',x_2') \vert \leq k\mbox{, then}\\
	& \vert b_{\gamma_1(e_i)}(f,x_1,x_2) \cap b_{\gamma_1(e_i)}(f',x_1',x_2') \vert = \vert b_{\gamma_2(e_i)}(\pi(f),x_1,x_2) \cap b_{\gamma_2(e_i)}(\pi(f'),x_1',x_2') \vert
	\end{aligned}\\
	\mbox{otherwise}\\
	\phantom{\mbox{if }}
	\mbox{also }
	\vert b_{\gamma_2(e_i)}(\pi(f),x_1,x_2) \cap b_{\gamma_2(e_i)}(\pi(f'),x_1',x_2') \vert > k.
	\end{cases}
	\]
\end{itemize}

Roughly speaking, the first condition guarantees that when extending \(\cH_i\) by \(\{e_{i+1}, \dotsc, e_k\}\)-partition equivalent drawings of \(e_i\), the topological separation of all vertices that might be important when drawing \(e_{i+1}, \dotsc, e_k\) is the same.
The second condition ensures that when extending \(\cH_i\) by \(\{e_{i+1}, \dotsc, e_k\}\)-partition equivalent drawings of \(e_i\), the number of edges whose drawings might be crossed by drawings of \(\{e_{i+1} \dotsc, e_k\}\) is the same, or so large that they cannot all be crossed by drawings of \(\{e_{i+1} \dotsc, e_k\}\).
\iflong{
This is more formally captured and used in the proof of the following lemma.}
\begin{lemma}
\label{lem:hidepart_-equiv}
	For any \(1 \leq i \leq k\),
	if two drawings \(\gamma_1(e_i), \gamma_2(e_i)\) of \(e_i\) into a drawing \(\cH_i\) of \(H_i\) are \(\{e_{i+1}, \dotsc, e_k\}\)-partition-equivalent, they either both can be extended to a 1-planar drawing of \(G\), or none of them can.
\end{lemma}
\iflong{
\begin{proof}
	We show that we can obtain a 1-planar drawing extension of \(\cH_i + \gamma_2(e_i)\) to \(G = H_i + \{e_{i+1}, \dotsc, e_k\}\) from a 1-planar drawing extension of \(\cH_i + \gamma_1(e_i)\) to \(G\).
	Then the claim immediately follows by a symmetric argument when \(\gamma_1(e_i)\) and \(\gamma_2(e_i)\) are interchanged.
	
	Let \(\pi\) be a bijection between the cells of \(\cH_i + \gamma_1(e_i)\) and the cells of \(\cH_i + \gamma_2(e_i)\) that witnesses \(\{e_{i+1}, \dotsc, e_k\}\)-partition equivalence of \(\gamma_1(e_i)\) and \(\gamma_2(e_i)\).
	Assume we are given a 1-planar drawing extension \(\cG_1\) of \(\cH_i + \gamma_1(e_i)\) to \(G\).
	From this, we will define a 1-planar drawing extension \(\cG_2\) of \(\cH_i + \gamma_2(e_i)\) to \(G\).
	For \(e \in E(H_i)\) set \(\cG_2(e) = \cH_i(e)\) and set \(\cG_2(e_i) = \gamma_2(e_i)\).
	In this way, \(\cG_2\) is an extension of \(\cH_i\).
	
	Note that for any cell \(f\) of \(\cH_i + \gamma_1(e_i)\) the order in which the vertices of \(X\) occur on the boundary of \(f\) is the same (up to possibly reversal) in which they occur on the boundary of \(\pi(f)\)
	(exactly the same such vertices occur because of \(\{e_{i+1}, \dotsc, e_k\}\)-partition equivalence).
	This is due to the fact that \(\cH_i + \gamma_1(e_i)\) and \(\cH_i + \gamma_2(e_i)\) are obtained from the same drawing \(\cH_i\) and drawing edges into \(\cH_i\) merely subdivides cells and cannot permute the order on their boundaries.
	
	Now we can define \(\cG_2(e_j)\) for \(j \in \{i + 1, \dotsc, k\}\) as follows:
	For \(J \subseteq \{i + 1, \dotsc, k\}\) such that \(\cG_1(e_j)\) intersects two cells \(f\) and \(g\) of \(\cH_i + \gamma_1(e_i)\) for every \(j \in J\), it holds that each \(\cG_1(e_j)\) crosses the drawing \((\cH_i + \gamma_1(e_i))(c_j)\) of an edge \(c_j \in E(H_i) \cup \{e_i\}\).
	In particular, \(c_j\) lies on the shared boundary of \(f\) and \(g\).
	Both \(f\) and \(g\) contain a vertex in \(X\) in their boundary, as each of them contain at least one endpoint of \(e_j\).
	Hence there are \(x_1, x_2 \in X\) that are consecutive on the boundary of \(f\) neglecting everything but \(X\), and \(y_1, y_2 \in X\) that are consecutive on the boundary of \(g\) neglecting everything but \(X\) such that \(c_j \in b_{\gamma_1(e_i)}(f,x_1,x_2) \cap b_{\gamma_1(e_i)}(g,y_1,y_2)\).
	By partition-equivalence the boundaries of \(\pi(f)\) and \(\pi(g)\) each contain an endpoint of each \(e_j\),
	and because \(|J| \leq k\), we find distinct \(c_j' \in b_{\gamma_2(e_i)}(\pi(f),x_1,x_2) \cap b_{\gamma_2(e_i)}(\pi(g),y_1,y_2)\) (or possibly \(c_j' \in b_{\gamma_2(e_i)}(\pi(f),x_2,x_1) \cap b_{\gamma_2(e_i)}(\pi(g),y_2,y_1)\)) for each \(j \in J\).
	Without loss of generality the \(c_j\) are indexed in the order in which they occur on the counterclockwise \(x_1\)-\(x_2\)-path along the boundary of \(f\).
	We re-index the \(c_j'\) to conform to the same order (up to reversal), also taking \(x_1\) and \(x_2\) into account, on \(\pi(f)\).
\end{proof}
}

The next lemma shows that the number of non-equivalent drawings is bounded by a function of $k$, which in turn allows us to apply exhaustive branching to prove the theorem.
\begin{lemma}
\label{lem:part_-equivbound}
	For any \(1 \leq i \leq k\),
	the number of ways to draw \(e_i\) into a drawing \(\cH_i\) of \(H_i\) that are pairwise not \(\{e_{i+1}, \dotsc, e_k\}\)-partition-equivalent is at most
	\(4(2k + 1) \cdot 2(k + 1) \in \bigoh(k^2)\).
\end{lemma}
\iflong{
\begin{proof}
	An equivalence class in question is, by definition, determined by
	a partition of the vertices of \(X\) on the boundaries, that a drawing of \(e\) in this equivalence class induces, and the number of uncrossed edges on the shared boundary of pairs of faces that are involved in this new partition.
	
	As the endpoints of \(e\) and their drawings are determined, two of three possible points at which a drawing of \(e\) partitions boundaries of faces of \(\cH_i\) are fixed.
	The possible third point lies between two of at most \(2k\) consecutive vertices in \(X\).
	This gives us \(2k + 1\) options for a possible third point (including the option not to have a third point).
	From this point we can reach each endpoint by a curve in counterclockwise or clockwise direction.
	
	Similarly, once one fixes the induced partitions, the impact a drawing \(\gamma(e)\) of \(e\) into \(\cH_i\) on the possible sizes of the face of the boundary between two vertices in $ X $ is quite restricted as it only impacts the values for adjacent cells of \(\cH_i + \gamma(e)\) that are bounded by parts of \(\gamma(e)\)  and the previously consecutive \(x_1, x_2 \in X\) that are partitioned by \(\gamma(e)\).
	These are at most two pairs, where the value for one pair implies the value for the other.
	Thus it suffices to distinguish which pair has the smaller value and what this value is among \(\{1, \dotsc, k, >k\}\), which results in \(2(k + 1)\) many possibilities.
\end{proof}
}

\fptkdirect

\iflong{
\begin{proof}
	We can pre-compute the intersection of the boundary of each cell of \(\cH\) with \(X\) and
	for each pair of cells \(f, f'\) of \(\cH\) and ordered pairs of vertices \(x_1, x_2 \in X\) and \(x_1', x_2'\) that are consecutive on the boundaries of \(f\) and \(f'\) respectively if one neglects everything but \(X\), the cardinality of the set of edges that are on the clockwise $x_1$-$x_2$-path along the boundary of $f$ and at the same time on the clockwise $x_1'$-$x_2'$-path along the boundary of $f'$
	in polynomial time.
	
	As described in the proof of Lemma~\ref{lem:part_-equivbound},
	at any stage, for \(1 \leq i \leq k\),
	we can branch on \(\{e_{i+1}, \dotsc, e_k\}\)-partition-equivalent drawings \(\gamma(e)\) of \(e\) using the pre-computed information.
	This information can be modified within each branch according to the choice of \(\gamma(e)\) in constant time because, as described in the proof of Lemma~\ref{lem:part_-equivbound} the impact of \(\gamma(e)\) involves only few values whose modifications can correctly be computed from the updated pre-computed information up to this stage and the chosen values determining \(\gamma(e)\).
	Correctness of this branching follows from Lemma~\ref{lem:hidepart_-equiv}.
\end{proof}
}

\section{Using Vertex+Edge Deletion Distance for \textsc{IC-Planar Drawing Extension}}
\label{sub:icfpt}
\label{sec:kappa}
\newcommand{\rreg}{{region}}
In this section, we show that
\textsc{IC-Planar Drawing Extension} parameterized by $\kappa$ is fixed-parameter tractable. 
We note that an immediate consequence of this is the fixed-parameter tractability of \textsc{IC-Planar Drawing Extension} parameterized by $k$.

On a high level, our strategy 
is similar to the one used to prove Theorem~\ref{thm:1Pfpt}, in the sense that we also use a (more complicated) variant of the patterns 
along with Courcelle's Theorem. However, obtaining the result 
requires us to extend the previous proof technique 
to accommodate the fact that
the number of edges incident to $\aV$, and hence the size of a pattern, is no longer bounded by $ \kappa $.
This is achieved by identifying so-called \emph{difficult vertices} and \emph{regions} that split up the neighborhood of each face-vertex in the \edg\ into a small number of sections (a situation which can then be handled by a formula in Monadic Second Order logic). Less significant complications are that we need a stronger version of Lemma~\ref{lem:notfar} to ensure that the diameter of the resulting graph is bounded, and need to be more careful when using MSO logic in the proof of the main theorem.

Let $f$ be a face of $\pH$ and let $\cG$ be a solution (i.e., an IC-planar drawing of $G$) for the instance $(G,H,\cH)$. Let $\eH$ be the \edg\ of $\cH$, and without loss of generality let us assume (via topological shifting) that each edge between a vertex $a'$ on the boundary of $f$ and a vertex $b\in \aV$ placed by $\cG$ in $f$ is routed ``through'' one shadow copy of $a'$\footnote{The reason one distinguishes which shadow copy of $a'$ the edge is routed through is because this unambiguously identifies which part of the face the edge uses to access $a'$.}. Let $\aV^f$ be the subset of $\aV$ drawn by $\cG$ in the face $f$.

Observe that, since shadow vertices are not part of the original instance and instead merely mark possible ``parts'' of the face that can be used to access a given vertex, it may happen that a solution routes several edges through one shadow vertex.
We say that a shadow vertex $v\in N_{\eH}(v_f)$ (where $N_{\eH}(v_f)$ denotes the neighborhood of $v_f$ in $\eH$) 
is \emph{difficult} w.r.t.\ $f$ if $\cG$ routes at least two edges through $v$. Note that it may happen that a vertex $v'\in \iV$ with more than one neighbor in $\aV$ has several shadow copies, none of which are difficult (see Figure~\ref{fig:not_difficult}).

\begin{figure}
	\begin{center}
		\includegraphics[page=2]{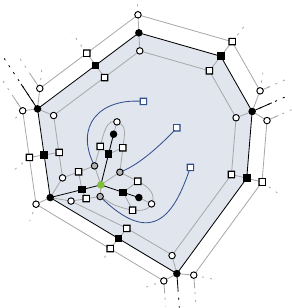}
	\end{center}
	\caption{An example of $\eH$ where a vertex $v'\in \iV$ has several non-difficult shadow copies. Blue vertices are in $ \aV $. The green vertex has no difficult shadow copy w.r.t.\ the blue face in $ \pH $.
		\label{fig:not_difficult}
	}	
\end{figure}


\newcommand{\regbound}{3\kappa}
\newcommand{\diffbound}{3\kappa^2}

	\begin{lemma}
\label{lem:diffbound}
There are at most $ \diffbound $ difficult vertices w.r.t.\ a face $ f $ of $ \pH $.
\end{lemma}
\begin{proof}
	We show that any two of the \(\ell\) added vertices drawn into \(f\) in \(\cG\) are both connected to at most \(3\) vertices in \(N_{\eH}(v_f)\).
	Then the claim follows.
	Assume for contradiction that \(v_1, v_2 \in \aV\) are drawn into \(f\) in \(\cG\) and \(w_1, w_2, w_3, w_4 \in N_{\eH}(v_f)\) are shadow vertices that each route two edges,
	one of which is incident to \(v_1\) and one of which is incident to \(v_2\).
	Since \(H\) is connected, the boundary of \(f\) is connected, and by construction of \(\eH\), \(w_1, w_2, w_3, w_4\) all lie on a cycle in \(\eH\) that does not involve any of \(\{v_{f'} \mid f' \text{ face in } \pH\}\).
	Hence the following graph \(H'\) is a minor of \(\eH - \{v_{f'} \mid f' \text{ face in } \pH\} + \aV + \aE\):
	\(H' = (\{v_1, v_2, w_1, \dotsc, w_4, v_f\} ,\{v_iw_j \mid i \in \{1,2,f\}, j \in \{1, \dotsc, 4\} \cup \{w_iw_{(i \operatorname{mod} 4) + 1} \mid i \in \{1, \dotsc, 4\})\).
	\(H'\) does not admit a 1-planar drawing in which both \(v_1\) and \(v_2\) lie on the same side of the drawing of the \(w_1\)-\(w_2\)-\(w_3\)-\(w_4\)-cycle and \(v_1\) and \(v_2\) are each incident to at most one edge whose drawing is crossed.
	However the existence of \(\cG\) implies that exactly such a drawing of \(H'\) exists.
\end{proof}

\ifshort{
A \rreg\ $R$ of a vertex $x\in \aV^f$ (or, equivalently, of a face $f$) is a maximal path $(r_1,\dots, r_p)$ in $N_{\eH}(v_f)$ with the following properties: (1) $\cG$ does not route through any shadow copy of an edge in $R$; (2) for each vertex $r$ in $R$, a uncrossed curve can be drawn in $\cG$ inside $f$ between $r$ and $x$; (3) none of the vertices in $R$ are adjacent to $\aV^f\setminus \{x\}$; and  (4) \(r_1\) and \(r_p\) are adjacent to \(x\).
}
\iflong{
A \rreg\ $R$ of a vertex $x\in \aV^f$ (or, equivalently, of a face $f$) is a maximal path $(r_1,\dots, r_p)$ in $N_{\eH}(v_f)$ with the following properties:
\begin{itemize}
\item $\cG$ does not route through any shadow copy of an edge in $R$;
\item for each vertex $r$ in $R$, a non-crossing curve can be drawn in $\cG$ inside $f$ between $r$ and $x$;
\item none of the vertices in $R$ are adjacent to $\aV^f\setminus \{x\}$;
\item \(r_1\) and \(r_p\) are adjacent to \(x\).
\end{itemize}
}

	\begin{lemma}
\label{lem:regbound}
There are at most \(\regbound\) \rreg s of a face $f$.
\end{lemma}
\begin{proof}
	Consider a path \(P\) in \(\eH[N_{\eH}(v_f)]\) that traverses all \(\ell\) regions of a vertex \(v \in \aV^f\).
	It contains at least \(\ell - 1\) pairwise disjoint subpaths \(P_1, \dotsc, P_{\ell - 1}\) of paths connecting regions of \(v\) that are consecutive in \(P\).
	
	For every \(P_i\) (\(i \in \{1, \dotsc, \ell\}\)), by the property that regions are inclusion maximal paths of vertices with certain properties, we find some vertex \(x\) in \(P_i\) that has to violate one of these properties.
	This can happen in three ways:
	(1) \(x\) is a shadow copy of an edge and $\cG$ routes through \(x\),
	(2) (1) is not the case and the drawing of some edge \(e \in \aE\) separates \(x\) from \(v\) in \(\cH\), or
	(3) (1) and (2) are not the case and \(x\) is adjacent to another vertex \(w \in \aV^f \setminus \{v\}\).

	In case (1) there is an edge \(e\in\aE\) such that the drawing of $e$ crosses the boundary of \(f\) in \(\cH\) and routes through \(x\).
	
	In case (2) the edge in question has both endpoints on \(P\), thus \(e \in \aEpar\) or the drawing of \(e\) crosses either the boundary of \(f\) in \(\cH\) or a drawing of another added edge.
	
	There are at most \(|\aEpar|\) edges in \(\aEpar\) and at most \(|\aV|\) edges in \(\aE \setminus \aEpar\) that can cross another edge in an IC-planar drawing.
	Moreover, in both cases (1) and (2), the endpoints of \(e\) cannot occur in any \(P_j\) with \(j \in \{1, \dotsc, \ell\} \setminus \{i\}\).
	
	In case (3) either \(x\) is contained in a region of \(w\) or the drawing of \(xw\) in \(\cH\) crosses an edge and is the only drawing of an edge incident to \(w\) that does so.
	If \(x\) is in a region of \(w\), \(w\) is separated from \(N_{\eH}(v_f) \setminus P_i\) by the edges from \(v\) to the outermost vertices of the regions that \(P_i\) connects and hence can have no region outside of \(P_i\).
	There are at most \(|\aV|\) many such \(w\).
	
	Thus we find at most \(|\aEpar| + |\aV| + 2|\aV|\) such \(x\) on \(P\) in total and thus 
	\(\ell \leq |\aEpar| + |\aV| + 2|\aV| \leq 3\kappa\).
	This concludes the proof.
\end{proof}
\newcommand{\vertIV}{\textnormal{vertex}}

The underlying intuition one should keep about regions and difficult vertices is that a solution $\cG$ partitions the shadow vertices into those which \textbf{(a)} have no edges routed through them, \textbf{(b)} have precisely one edge routed through them (in which case they must be part of the respective region), and \textbf{(c)} have at least two edges routed through them (in which case they form a difficult vertex). 
\iflong{Next, we extend the notion of a \emph{pattern} from Definition~\ref{def:pat}.}
\iflong{One technical distinction is that instead of using cyclically ordered multisets for $C(s)$, we use \emph{cyclic orders with equivalences} where two elements of the multiset can be assigned to the same position in the cyclic order (i.e., they can be equivalent).
\begin{definition}
\label{def:epat}
An \emph{extended pattern} is a tuple $(S,Q_1,Q_2,C)$ where 
\begin{enumerate}
\item $S$ is a set of at most $2k$ elements;
\item $Q_1$ is a mapping from $\aV$ to $S$;
\item $Q_2$ is a mapping from $\aV\cup \aEpar$ to totally ordered subsets of $S\cup\{\times\}$ of cardinality $1, 2$ or $3$ such that, for each vertex $v\in \aV$, the first element of $Q_2(v)$ is $Q_1(v)$;
\item Let $Q^{-1}_1(s)$ be the set of vertices in $\aV$ mapped by $Q_1$ to $s\in S$, let $M=2^{Q^{-1}_1(s)}$, and let $Q^{-1}_2(s)$ be the set of vertices in $\aV$ and edges in $\aEpar$ mapped by $Q_2$ to $s$. Then $C$ is a mapping from $S$ that maps each $s\in S$ to a multiset that is cyclically ordered with equivalences, containing:
\begin{itemize}
\item at most $\regbound$ elements of $Q^{-1}_1(s)$, 
\item at most $\diffbound$ elements of $M$, 
\item for each $x\in \aV\cap Q^{-1}_2(s)$, at most $3$ pair of the form $(x,q)$ where $q\in \{\cros,\vertIV\}$,
\item for each $ab\in \aEpar\cap Q^{-1}_2(s)$, at most $4$ pairs of the form $(ab,q)$ where $q\in \{\cros,a,b\}$. 
\end{itemize}

\end{enumerate}
\end{definition}
}

\iflong{
Unlike the patterns used to prove Theorem~\ref{thm:1Pfpt}, extended patterns do not track the exact placement of vertices in $\iV$ (since $|\iV|$ is not bounded by $\kappa$). To make up for this, the cyclic orders stored in $C$ detail the order in which individual regions, difficult vertices together with crossings and endpoints of edges in $\aEpar$ and the single ``special'' edge per vertex in $\aV$ that is allowed to cross, are supposed to appear inside a face. The element $\times$ is used by $Q_2$ to capture whether an edge (either denoted explicitly by $e$, or representing the special ``potentially crossing edge'' of $v\in \aV$) is crossing.}

\ifshort{%
	In the remainder of this section we 
	give some intuition on 
	how the techniques used to prove Theorem~\ref{thm:1Pfpt} need to be extended
	to obtain Theorem~\ref{thm:ICPfpt}.
	Unlike in the proof of Theorem~\ref{thm:1Pfpt} the number of vertices in $ \iV $ is not bounded by $ \kappa $.
	Consequently, we cannot track their exact placement through patterns as defined in Definition~\ref{def:pat}.
	To circumvent this, we define \emph{extended patterns} that store the necessary details about
	the cyclic orders in which individual regions, difficult vertices together with crossings, and endpoints of edges in $\aEpar$ as well as the single edge per vertex in $\aV$ that is allowed to cross, are supposed to appear inside a face. 
	Then, as for patterns, we define how an extended pattern is \emph{derived} from potential solutions.}

\iflong{Notice that this information was not required for solving \textsc{1-Planar Drawing Extension}, but for IC-planarity we need to ensure (via the MSO formula employed in the proof of Theorem~\ref{thm:ICPfpt}) that each vertex is incident to at most $1$ crossing edge.}
\iflong{

Equivalences in the cyclic order are used when a single vertex is simultaneously a difficult vertex inside a face but also an endpoint of one or several crossing edges; we note that the use of equivalences in the cyclic order could be avoided for simple patterns since parameterizing by $k$ allowed us to explicitly refer to individual added edges and their endpoints.
As in Subsection~\ref{subsec:1pfpt}, we proceed by defining a notion of validity and pattern graphs for our extended patterns, and these may also provide further intuition for what information is carried by an extended pattern.

\begin{definition}
\label{def:derivedic}
For a solution $\cG$, we define a derived pattern $P=(S,Q_1,Q_2,C)$ as follows:
\begin{itemize}
\item $S$ is the set of faces of $\pH$ which have a non-empty intersection if $\cG(e)$ for some $e\in \aE$. 
\item For $v\in \aV$ we set $ Q_1(v)$ to the face \(f\) of \(\pH\) for which $ \cG(v) $ lies inside $f$.
\item For $e\in \aEpar$ we set $Q_2(e) $ to the set of at most two faces which have a non-empty intersection with $\cG(e)$.
\item For $v\in \aV$, if $v$ is incident to an edge $e$ with a crossing, then we set $Q_2(v)$ to the set of at most two faces which have a non-empty intersection with $\cG(e)$, else we set $Q_2(v)$ to the face \(f\) of \(\pH\) for which $ \cG(v) $ lies inside $f$.
\item For a face $ s \in S $ we consider all the difficult vertices w.r.t.\ $f$, all the regions for vertices $v\in Q_1^{-1}(s)$, and 
all the edges $ e = uv \in \aE $ that either cross or have both endpoints in $\iV$ with a non-empty intersection between $ \cG(e) $ and $s$. 
For such edge $e$, there is an edge $ e' \in E(H) $ on the boundary of $s$ such that $ \cG(e) $ crosses $ \cG(e') $, or $ u \in \iV $ and $ u $ is on the boundary of $ s $, or both. 
We set $ C(s) $ as the ordered multiset of these difficult vertices, regions, and crossing points or vertices when traversing $ s $ in clockwise fashion.
Equivalent elements in the cyclic order are created when edges are routed through the same shadow vertex.
\end{itemize} 
\end{definition}

\begin{definition}\label{def:valid2}
An extended pattern $P$ is \emph{valid} if there exists a \emph{pattern graph} $G_P$ with an $IC$-planar drawing $\cG_P$ satisfying the following properties:
\begin{itemize}
\item $\aV \subseteq V(G_P)$.
\item $\dr{G}_P-\aE$ is a planar drawing.
\item $S$ is a subset of non-outer faces of $\cG_P-\aE$.
\item Each $v\in \aV$ is contained in the face $Q_1(v)$ of $\cG_P-\aV-\aE$.
\item Each $e\in \aEpar$ is contained in the face(s) $Q_2(e)$ of $\cG_P-\aE$.
\item For each edge $vw$ with $v,w\in \aV$, $vw$ is contained in the face(s) $Q_2(v)$ of $\cG_P-\aV-\aE$.
\item When traversing the inner side of the boundary of each face $s$ of $\cG_P-\aE$ in the clockwise fashion, the order in which difficult vertices, \rreg s, and edges $e\in \aE$ that either cross $s$ or are in $\aEpar$ together with the information whether $e$ crosses here or ends is precisely $C(s)$.
\end{itemize}
\end{definition}

Observe that for each solution, the derived pattern is valid by definition. We also note that the total number of extended patterns can be bounded in $\kappa$ in an analogous way as we bounded the number of patterns in Subsection~\ref{subsec:1pfpt}---notably, the number of extended patterns is at most $\numpat(\kappa)\in 2^{\bigoh(k^2\log k)}$.

\begin{lemma}
\label{lem:validcheck2}
	Given extended pattern $P = (S,Q_1, Q_2,C)$ and an instance $ (G,H,\mathcal H) $, in time $O((\kappa^2!)^{\kappa^2}*{\kappa^2}^{2\kappa^2 + 1})$ we can either construct a pattern graph $G_P$ together with the drawing $\cG_P$ satisfying all the properties of Definition~\ref{def:valid2} or decide that $P$ is not valid.
\end{lemma}
\iflong{
\begin{proof}
	The proof works very similar to the one of Lemma~\ref{lem:validpattern}. 
	In fact the main difference lies in the categories of vertices on the cycles $ \mathcal C(s) $ we construct for each $ s \in S $ from $ C(s) $. 
	Previously these contained vertices $ v_i^s $ representing either a vertex $ v \in \iV $ or a crossing. 
	Since $ |\iV| $ is not bounded by $ \kappa $, we now have to differentiate between vertices $ v_i^s $ representing either a difficult vertex, a region, a vertex incident to an edge in $ \aEpar $, or a crossing. 
	
	More formally, we introduce a vertex $ v_i^s $ for each element $ c_i \in C(s) $ for every $ s \in S $ and connect the $ v_i^s $ to a cycle in the order implied by $ C(s) $.
 	For each such $ v_i^s $ we remember the $ c_i $ it represents.
	Every vertex $ u \in \aV $ we represent by a vertex $ u_s $. 
	One $ u_s $ we connect to all vertices $ v_i^s $ for which either $ u_s $ is in the set of vertices $ c_i $ represented by $ v_i^s $, or, if $ v_i^s $ represents a tuple $ (z,q) $, we connect $ u_s $ to $ v_i^s $ whenever $ z = u $ or $ u $ is incident to the edge $ z $.
	Finally we identify all vertices $ v_i^s,v_j^{s'} $ which represent tuples $ (z_i,q_i) $ and $ (z_j,q_j) $ with $ q_i = q_j = \cros $ and $ z_i = z_j $.
	Guessing the crossings between edges incident to $ v_i^s $s can be done as in proof of Lemma~\ref{lem:validpattern}. 
	Using Lemma~\ref{lem:diffbound} and~\ref{lem:regbound} the resulting graph $ G_P' $ has bounded number of vertices in the parameter, i.e., $ |V(G_P')| = O(\kappa^2) $. 
	Hence we can test all possible drawings and return such a planar drawing $ \cG_P $ with $ G_P $, after checking the conditions of Definition~\ref{def:valid2} and introducing the crossings as in Lemma~\ref{lem:validpattern}.
		
	The reverse direction can be proven as above.
\end{proof}
}

\begin{lemma}
\label{lem:notfaric}
Let $I=(G,H,\cH)$ be an instance of \textsc{IC-Planar Drawing Extension}.
Let $Z$ be the set of all vertices in $\eH$ of distance at least $4\kappa+7$ from each vertex in $\iV$. Let $G'$, $H'$, and $\cH'$ be obtained by deleting all vertices in $Z$ from $G$, $H$, and $\cH$.
Then:
\begin{enumerate}
	\item If $I$ is a \textsc{YES}-instance, then each connected component of $G'$ contains at most one connected component of $H'$;
	\item $I$ is a \textsc{YES}-instance if and only if for each connected component $A$ of $H'$ the restriction of $\cH'$ to $H'[A]$ can be extended to a drawing of the connected component of $G'$ containing~$A$. Moreover, given such IC-planar extension for every connected component of $G'$, there is an algorithm that outputs a solution for $I$ in linear time.
	\item either for each connected component $A$ of $H'$ the \edg\ of $\cH'[A]$ has diameter at most $\bigoh(\kappa^2)$, or $I$ is a \textsc{NO}-instance.
\end{enumerate}
\end{lemma}
\iflong{
The proof of the lemma follows the same general strategy as our proof of Lemma~\ref{lem:notfar}, with Point 3. borrowing some ideas from the proof of Lemma~\ref{lem:gettw}. We split the proof of Lemma~\ref{lem:notfar} into proofs for the two individual points.

\begin{proof}[Proof of Point 1]
	For the sake of contradiction let $J$ be a connected component of $G'$ that contains two distinct connected components $H_1'$ and $H_2'$ of $H'$. Since $J$ is a connected component, there must be a path $P$ from a vertex $v_1\in H_1'$ to a vertex $v_2\in H_2'$ in $J-(H_1'\cup H_2')$, and moreover $P$ must have length at most $|\aV|+1$. By definition, both $v_1$ and $v_2$ are in $\iV$. To complete the proof, it suffices to show that in any solution $\cG$, $v_1$ and $v_2$ have distance at most $8\kappa+8$ in $\eH$ and hence all the vertices and faces on the shortest $v_1$-$v_2$ in $\eH$ are at distance at most $4\kappa+4$ from a vertex in $\iV$ and remain unchanged in $\cH'$.
	
	
	Moreover, in any solution $\cG$, two consecutive vertices of $P$ are either drawn in the same face of $\pH$ or in two adjacent faces of $\pH$. Observe that the distance in $\eH$ between two face-vertices for the faces that share an edge is $4$, and that the distance from an original vertex $v$ to a face-vertex of a face incident to $v$ is $2$. Therefore, if $(G,H,\cH)$ is a \textsc{YES}-instance, then the distance between $v_1$ and $v_2$ in $\eH$ must be at most $4|\aV|+8\le 4\kappa+8$.
\end{proof}

\begin{proof}[Proof of Point 2]
	The forward direction is obvious. For the backward direction, let $G_1,\ldots, G_r$ be the connected components of $G'$ and for $i\in [r]$ let $H_i$ and $\cH_i$  be the restriction of $H'$ and $\cH'$, respectively, to $G_i$. Moreover, let $\pH_i$ be the planarization derived from $\cH_i$ and note that $H_i$ is connected for all $i\in [r]$ by Point 1. 
	Now let us fix an arbitrary $i\in [r]$ such that $H_i$ is not empty and let $\cG_i$ be a 1-planar extension of $\cH_i$ to $G_i$. 
	
	Observe that each face of $\pH$ is completely contained in precisely one face of $\pH_i$.
	Moreover, if a face $f$ of $\pH_i$ contains at least two faces $f_1$ and $f_2$ of $\pH$, then both $v_{f_1}$ and $v_{f_2}$ are at distance at least $4\kappa+4$ of any vertex in $\iV\cap V(H_i)$ in $\eH$. Indeed, if this were not the case, then w.l.o.g.\ the vertices on the boundary of $v_{f_1}$ would have distance at most $4\kappa+6$ from some $w\in \iV\cap V(H_i)$ in $\eH$, which would mean that $f_1$ is also a face in $\pH_i$. By the same distance-counting argument introduced at the end of the Proof of Point 1, This implies that no edge in a path $P$ of $G$ from a vertex $v\in H_i$ whose internal vertices all lie in $\aV$ can be drawn in any face of $\pH$ contained in $f$.
	
	To complete the proof, let $G_1, \ldots, G_p$, $p\le r$ be the connected components of $G'$ that contain a vertex in $H$ and $G_{p+1},\ldots, G_r$ the remaining connected components of $G'$. We obtain a solution $\cG$ to the instance $I$ by simply taking the union of $\cH$ and $\cG_i$ for $i\in [p]$ and then for $i\in \{p+1,\ldots, r \}$ shifting $\cG_i	 $ so that $\cG_i$ do not intersect any other part of the drawing. Note that vertices in $\iV$ in different connected components are far apart and hence this union cannot introduce a vertex incident to two crossing edges. 
\end{proof}

\begin{proof}[Proof of Point 3]
Now let us consider a connected component $A$ of $H'$, let $\eH_A$ be  the \edg\ of $\cH'[A]$, $\pH_A$ the planarization of $\cH'[A]$, and let $v_f$ be a face-vertex in $\eH_A$. If $v_f$ is at distance at least $4\kappa+9$ from every vertex in $\iV\cap A$ in $\eH_A$, then every vertex on the boundary of $f$ is at distance at least $4\kappa+7$ from every vertex $w\in \iV\cap A$ in $\eH_A$. Let $v$ be an arbitrary vertex incident to $f$ in $\cH'[A]$. Since each face of $\pH$ is completely contained in precisely one face of $\pH_A$, it follows that $v$ is at distance at least $4\kappa+7$ from each vertex $w\in\iV\cap A$ in $\eH$. Because $v\in V(H'[A])$, this contradicts the fact that every vertex in $V(H')$ is at distance at most $4\kappa+6$ from a vertex $w\in\iV$ in $\eH$. Hence, every face-vertex in $\eH_A$ is at distance at most $4\kappa+8$ from a vertex in $\iV\cap\cC$. To finish the proof it suffice to show that if $(G'[A],H'[A],\cH'[A])$ is \textsc{YES}-instance, then there exists a set $C$ of at most $2\kappa$ face-vertices in $\eH_A$ such that every vertex in $\iV\cap A$ is at distance at most $6$ from a face-vertex in $C$. Now let us consider a solution $\cG'[A]$ to $(G'[A],H'[A],\cH'[A])$ and let $C$ be the set of face-vertices of faces that either contain a vertex in $\aV$ or intersect an edge in $\aEpar$. Clearly, the size of $C$ is at most $2\kappa$. Now each vertex in $\iV$ is either incident to an edge in $\aEpar$, in which case it is incident to some face that $\aEpar$ intersects, or it is adjacent to a vertex in $\aV$. In the second case it is either incident to the face containing its neighbor in $\aV$ or it is incident with a face that have a common edge with the face containing its neighbor in $\aV$. It follows that every vertex in $\eH_A$ is at distance at most $4\kappa+15$ from a vertex in $C$ and the diameter of $\eH_A$ is at most $4\kappa\cdot(4\kappa+15)=\bigoh(\kappa^2)$.
\end{proof}	
}
\begin{lemma}
\label{lem:gettwic}
\textsc{IC-Planar Drawing Extension} is \FPT\ parameterized by $\kappa+\tw(\eH)$ if and only if it is \FPT\ parameterized by $\kappa$.
\end{lemma}
\iflong{
\begin{proof}
The backward direction is trivial. For the forward direction, assume that that there exists an algorithm $\cB$ which solves \textsc{IC-Planar Drawing Extension} in time $f(\kappa+\tw(\eH))\cdot |V(G)|^c$ for some constant $c$ and computable function $f$. Now, consider the following algorithm $\cA$ for \textsc{IC-Planar Drawing Extension}: $\cA$ takes an instance $(G_0,H_0,\cH_0)$ and constructs $(G_1,H_1,\cH_1)$ by applying Lemma~\ref{lem:notfaric}. Recall that by Point 1 of Lemma~\ref{lem:notfaric}, $(G_0,H_0,\cH_0)$ is either \textsc{NO}-instance, in which case $\cA$ correctly outputs ``NO'', or each connected component of $G_1$ contains at most one connected component of $H_1$. 



Let us consider a connected component $\cC$ of $G_1$ and the \edg\ $\eH_1[\cC]$ of $\cH_1[\cC]$. By Point 3 of Lemma~\ref{lem:notfaric}, either the diameter, and in turn the radius and and by Proposition~\ref{prop:radius_treewidth} the treewidth, is bounded by $\bigoh(\kappa^2)$ or $\cA$ can correctly output ``NO''.
Now, for each connected component $\cC$ of $G_1$, we solve the instance $(G_1[\cC],H_1[\cC],\cH_1[\cC])$ using algorithm $\cB$. If $\cB$ determines that at least one such (sub)-instance is a \textsc{NO}-instance, then $\cA$ correctly outputs ``NO''. Otherwise, $\cA$ outputs a solution for $(G_0,H_0,\cH_0)$ that it computes by invoking the algorithm given by Point 2 of Lemma~\ref{lem:notfaric}. 
To conclude, we observe that $\cA$ is a fixed-parameter algorithm parameterized by $k$ and its correctness follows from Lemma~\ref{lem:notfaric}.
\end{proof}
}
}

\fptkappaic


\iflong{
\begin{proof}

	We prove the theorem by showing that \textsc{IC-Planar Drawing Extension} is fixed-parameter tractable parameterized by $\kappa+\tw(\eH)$, which suffices thanks to Lemma~\ref{lem:gettwic}.
	
	To this end, consider the following algorithm $\cA$. Initially, $\cA$ loops over all of the at most $\numpat(\kappa)$ many patterns, tests whether each pattern is valid or not using Lemma~\ref{lem:validcheck2}, and stores all valid patterns in a set $\calP$. Next, it branches over all valid patterns in $\calP$, and for each such pattern $P=(S=\{s_1,\dots,s_\ell\},Q_1,Q_2,C)$ it constructs an MSO formula $\Phi_P(\freeVar)$, where $\freeVar$ is a set of at most $6\kappa^2+7\kappa$ free variables specified later, 
	the purpose of which is to find a suitable ``embedding'' for $P$ in $\cH$ by finding an interpretation in the \edg\ $\eH$.

	In the following we will formally define the MSO formula $\Phi_P(\freeVar)$. Recall that the vertices of $\eH$ have the following labels: a label $v$ for every vertex $v\in \iV$ and then the labels $O$, $E$ ,$F$, $C$, $S$ which represent original, edge-, face-, crossing-, and shadow-vertices, respectively. However, since $\iV$ could be large, a label $v$ for every vertex $v\in \iV$ is not feasible to use Fact~\ref{fact:MSO}. Therefore, instead we will have label $w$ for each vertex $w\in \aV$ specifying that vertex $v$ is adjacent to $w$ in $G$. However, we keep the labels $v$ for vertices $v$ incident to edges in $\aEpar$. Note that vertex $v$ can have several different labels at the same time.  
	For vertices $x, y$, let $\edge(x,y)$ be a formula stating that $x$ and $y$ are adjacent vertices, and $\conn(x, y, X)$ a formula stating that there is a directed path\footnote{Recall that edges between shadow vertices are directed.} from $x$ to $y$ with all inner vertices in $X$. 
	
	The set of free variables $\freeVar$ of $\Phi_P(\freeVar)$ consists of:
	\begin{itemize}
		\item $x_1,\ldots, x_\ell$, where $x_i$ corresponds to a single element $s_i$ in $S$; 
		\item $y_1, \ldots, y_{k'}$, where $y_i$ corresponds to either to an edge in $H$ that is crossed by either an edge $e_i\in \aEpar$ or the unique crossing edge $e_i$ from a vertex $v_i\in \aV$---formally, either $(e_i,\cros)\in C(s_j)$ or $(v_i,\cros)\in C(s_j)$ for some $j\in [\ell]$ (Note that this $e_i$ could either cross from one face of $\pH$ to another, but also could cross an edge of $H$ that is incident to a single face in $\pH$);
		\item for each $i\in [\ell]$, we have $Z_1^i, \ldots, Z_{q_i}^i$ -- where $q_i\le |C(s_i)|$ and $Z_j^i$ correspond to $j$-th consecutive set of equal elements of $C(s_i)$ (after fixing some arbitrary first element in the cyclic ordering). Note that $Z_j^i$ are set variables and formula needs to distinguish whether $Z_j^i$ is a region, or some set containing difficult vertices, or an endpoints of an edge, or it is a crossing vertex. Since only equivalent elements of $C(s_i)$ are allowed to map to the same vertex, We will require that $Z_j^i$'s are all disjoint.
	\end{itemize}
	
	Note that $\ell\leq 2\kappa$, $k'\leq \kappa$, and the total number of variables of the form $Z_j^i$ is upper-bounded by $6\kappa^2+4\kappa$.
	
	The formula $\Phi_P(\freeVar)$ is then the conjunction of the following subformulas:
	
	\begin{enumerate}
		\item checkFaces$(\freeVar)$, which ensures that $x_i$'s are assigned to distinct face-vertices and is the conjunction of:
		\begin{itemize}
			\item $P_Fx_i$, for all $i\in[\ell]$ and 
			\item $x_i\neq x_j$ for all $1\le i < j\le \ell$;
		\end{itemize}
		\item checkEdges$(\freeVar)$, which ensures that $y_i$'s are assigned to distinct edge-vertices and is the conjunction of:
		\begin{itemize}
			\item $P_Ey_i$, for all $i\in[k']$ and 
			\item $y_i\neq y_j$ for all $1\le i < j\le k'$;
		\end{itemize}
		\item checkShadow$(\freeVar)$, which ensures that $Z_j^i$'s are assigned to sets of disjoint shadow-vertices that are all adjacent~to~$x_i$:
		\begin{itemize}
			\item for all $i\in [\ell]$ and $j\in [|C(s_i)|]$ we have: \\ $\forall z \left(Z_j^iz\rightarrow \left(P_Sz\wedge \edge(x_i,z)\right)\right)$, and
			\item $\forall z (\neg Z_{j}^iz \vee \neg Z_{j'}^{i'}z)$ for all $i,i'\in[n]$ and $j\in [q_i]$, $j'\in [q_{i'}]$ such that $(i,j)\neq (i',j')$;
		\end{itemize}
	\item checkRegion$(\freeVar)$, which ensures that if $Z_j^i$ corresponds to a region then its set of variables is consecutive and that every region is disjoint from any other $Z_{j'}^{i'}$:
	\begin{itemize}
		\item for all $i\in [\ell]$ and $j\in [q_i]$ such that $Z_j^i$ correspond to a region we have: \\ $\forall z_1, z_2 \left( (Z_j^iz_1\wedge Z_j^iz_2)\rightarrow (\conn(z_1,z_2,Z_j^i)\vee \conn(z_2,z_1,Z_j^i))\right)$;
	\end{itemize}
	\item checkNoRegion$(\freeVar)$, which ensures that if $Z_j^i$ does not correspond to a region then it contains precisely one variable:
	\begin{itemize}
		\item for all $i\in [\ell]$ and $j\in [q_i]$ such that $Z_j^i$ does not correspond to a region we have: \\ $\exists z \left(Z_j^iz\wedge\left(\forall x \left(Zx\rightarrow (x= z) \ \right) \right)\right)$
	\end{itemize}
		\item checkCrossings$(\freeVar)$, which ensures that the edge-vertex $y_p$,
		corresponding to an edge $e_p\in \aEpar$ or the unique edge incident with a vertex $v_p\in \aV$ crossing an edge in $H$ incident to faces $s_1, s_2$, is adjacent to a vertex in $Z_{j_1}^{i_1}$ and $Z_{j_2}^{i_2}$ corresponding to the two pairs $(e_p,\cros)$ (or $(v_p,\cros)$) in $C(s_1)$ and $C(s_2)$, respectively: 
		\begin{itemize}
			\item for all $p\in [k']$ and the corresponding $Z_{j_1}^{i_1}$ and $Z_{j_2}^{i_2}$, checkCrossings$(\freeVar)$ contains: \\ $\exists z_1,z_2 \left( Z_{j_1}^{i_1}z_1 \wedge Z_{j_2}^{i_2}z_2\wedge \edge(z_{1}, y_p)\wedge \edge(y_p, z_2)\right)$.
		\end{itemize}
		\item check$\iV^1(\freeVar)$, which ensures that if incidence between an edge $e\in \aEpar$ and a vertex $v\in \iV$ is realized in the face $s_i$ (i.e., $(e,v)\in C(s_i)$), then the unique variable $z\in Z_j^i$ corresponding to $(e,v)$ in $C(s_i)$ is adjacent to $v$.
		\begin{itemize}
			\item For all $i\in [\ell]$ and all $(e,v)\in \aEpar\times\iV$ such that $(e,v)$ corresponds to $Z_j^i$, check$\iV^1(\freeVar)$ contains: \\ $\exists z,u \left(Z_j^iz\wedge P_vu\wedge \edge(x_i,z) \wedge \edge(, u)\right)$.
		\end{itemize}
	\item check$\iV^2(\freeVar)$, which ensures that if for $v\in \aV$ the set variable $Z_j^i$ corresponds to $(v,\vertIV)$ in $C(s_i)$, then the unique variable $z\in Z_j^i$ in $C(s_i)$ is adjacent to a neighbor of $v$.
	\begin{itemize}
		\item For all $i\in [\ell]$ and all $(v,\vertIV)\in \aV\times\{ \vertIV\}$ such that $(v,\vertIV)$ corresponds to $Z_j^i$, check$\iV^1(\freeVar)$ contains: \\ $\exists z, u \left(Z_j^iz\wedge P_vu\wedge \edge(x_i,z) \wedge \edge(z, u)\right)$.
	\end{itemize}
	\item checkIncidences$(\freeVar)$, which ensures that for a vertex $v$ all the neighbors of $v$ in $G$ are adjacent to some vertex in $Z_j^i$'s corresponding to regions, difficult vertices, or elements $(v,\vertIV)$. That is for every neighbor $x$ of $v$ we can draw an edge from $v$ to $x$ same way as in the pattern.
	\begin{itemize}
		\item Let $v\in\aV$ and let $Z_1^v,\ldots Z_r^v$ be all the $Z_j^i$'s that represent either a region of vertex $v$, difficult vertex w.r.t. some face $s_i$ with $v$ being in the corresponding set in $C(s_i)$ or an element $(v,\vertIV)$, then we have: 
		$\forall x \big(P_vx\rightarrow ((\exists y (Z_1^vy\wedge \edge(x,y)))\vee\ldots$ $ \vee (\exists y (Z_r^vy\wedge \edge(x,y)))) \big)$
	\end{itemize}
		\item checkCyclicOrder$(\freeVar)$, which ensures that $Z_j^i$'s occur in the cyclic order around the face-vertex $x_i$ given by $C(s_i)$:
		\begin{itemize}
			\item for all $i\in [\ell]$ and $j\in [|C(s_i)|]$ checkCyclicOrder$(\freeVar)$ contains:\\ 
			$\exists X,z_1,z_2 \big(Z_{j}^iz_1\wedge Z_{j+1}^iz_2\wedge (\forall x (Xx \rightarrow (\edge(x_i,x)\wedge \neg Z_1^ix \wedge\ldots$ $\wedge \neg Z_{|C(s_i)|}^ix ))$ $ \wedge \conn(z_1,$ $ z_2, X)) \big)$, where $Z_{|C(s_i)|+1}^i = Z_{1}^i$.
		\end{itemize}
	\item checkIC$(\freeVar)$ that verifies that every original vertex is incident to at most one crossing edge. That is for an original vertex we need to go through all incident edges in $\cH$ and edges to vertices in $\aV$ and check if they are crossing. The edge in $\cH$ is crossing either already in $\cH$, if it is one of edges $y_1$. The edge $e\in \aEpar$ is crossing if $\times\in Q_2(e)$, similarly the edge from $v\in \aV$ is crossing if $\times\in Q_2(e)$, $(v,\vertIV)$ is some $C(s_i)$ and $Z_j^i$ corresponding to $(v,\vertIV)$ is interpreted as a shadow copy of the given original vertex.
	\begin{itemize}
		\item 
		Let $Z_1,\ldots Z_r$ be $Z_j^i$'s corresponding to the entries $(ab,a)$, $(ab,b)$, and $(v,\vertIV)$ for edges $ab\in \aEpar$ and vertices $v\in \aV$ with $\times \in Q_2(ab)$ and $\times \in Q_2(v)$, respectively. Then we have the formula: \\
		$\forall x,y,z (\left(P_Ox\wedge (y\neq z)\wedge \edge(x,y)\wedge \edge(x,z) \right)\rightarrow$ 
		$(\neg ( (y=y_1)\vee \ldots \vee (y=y_k)\vee P_Cy\vee$ $Z_1y\vee \ldots\vee Z_ry)  ) \vee \neg \left( (z=y_1)\vee \ldots \vee(z=y_k)\vee P_Cz\vee Z_1z\vee \ldots\vee Z_rz\right) ) $.
	\end{itemize}
	\end{enumerate}
	
	Clearly, the length of the formula $\Phi_P(\freeVar)$ is bounded by a function of $\kappa$. 
	Hence, we can use Fact~\ref{fact:MSO} to, in time $f(\kappa,\tw(\eH))\cdot |\eH|$ for some computable function $f$, either decide that $\eH\not \models \Phi_P(\freeVar)$ or find an assignment  $\phi: \freeVar\rightarrow V(\eH)$ such that $\eH \models \Phi_P(\phi(\freeVar))$. 
	
	The rest of the proof now follows by repeating the arguments given in the proof of Theorem~\ref{thm:1Pfpt}---in particular, we will insert the IC-planar drawing $\cG_P$ of the pattern $G_P$ corresponding to $\Phi_P$ into the faces identified by the formula. The only substantial difference is that here, the pattern graph does not provide an explicit drawing for the edges between $\aV$ and a region on the face containing $\aV$---however, a drawing for these edges is easy to construct thanks to the existence of non-crossing curves connecting the region and the respective vertex in $\aV$.
\end{proof}
}

\section{Inserting Two Vertices into a 1-Plane Drawing}\label{sec:2vtcs}

In this section we show that \textsc{1-Planar Drawing Extension} is polynomial-time tractable in the case where we are only adding $2$ vertices to the graph along with their incident edges (i.e., when $ |\aV| = 2 $ and $ \aEpar =\emptyset$\footnote{We note that it is trivial to extend the result to the case where the number of added edges is bounded by a fixed constant, via simple exhaustive branching.}). Already solving this, at first glance simple, case seems to require non-trivial insight into the problem. In the following we call the two vertices in $ \aV $ the \emph{red} and \emph{blue} vertex,
denoted by $ r $ and $ b $, respectively.

On a high level, our algorithm employs a ``delimit-and-sweep'' approach. First, it employs exhaustive branching to place the vertices and identify a so-called ``initial delimiter''---a Jordan curve that isolates a part of our instance that we need to focus on. In the second step, it uses such an initial delimiter to solve the instance via a careful dynamic programming subroutine.
As our very first step, we exhaustively branch to determine which cells $r$ and $b$ should be drawn in, in \(\bigoh(n^2)\) time, and in each branch we add \(r\) and \(b\) into the selected cell(s) (from now on, we consider these embeddings part of $\cH$). 

\noindent \textbf{The Flow Subroutine.} \quad
Throughout this section, we will employ a generic network-flow subroutine that allows us to immediately solve certain restricted instances of \textsc{1-Planar Drawing Extension}. In particular, assuming we are in the setting where $r$ and $b$ have already been inserted into \(\cH\), consider the situation where:
\begin{itemize}
	\item There is a partial mapping \(\lambda\) from the faces of \(\cH^\times\) to \(\{R, B\}\);
	 and
	\item \(r\) and \(b\) are in different cells of \(\cH\).
\end{itemize}

We say a 1-planar extension of \(\cH\) to \(G\) is \emph{\(\lambda\)-consistent} if the drawing of any edge in \(E(G) \setminus E(H)\) which is incident to \(r\) intersects the interior of face \(F\) of \(\cH^\times\) only if \(\lambda(F) = R\), and correspondingly
the drawing of any edge in \(E(G) \setminus E(H)\) which is incident to \(b\) intersects the interior of face \(F\) of \(\cH^\times\) only if \(\lambda(F) = B\) (i.e., $\lambda$ specifies precisely which kind of edges may enter which face).\ifshort{ We use a reduction to network flows to show:
}
\iflong{

We show that for a given \(\lambda\) we can either find a \(\lambda\)-consistent extension or decide that there is none by constructing an equivalent network flow problem.
}


%

\begin{lemma}
	\label{lem:flow}
	Given $\lambda$ as above, it is possible to determine whether there exists a $\lambda$-consistent 1-planar extension of \(\cH\) to \(G\) in polynomial time.
\end{lemma}

\iflong{\begin{proof}}
\ifshort{\begin{proof}[Proof Sketch]}
	Consider the max flow instance $\theta_1$ constructed as follows. $\theta_1$ contains a universal sink $t$ and a universal source $s$. We add one vertex for each vertex in $N_{E(G) \setminus E(H)}(r)$, and a capacity-1 edge from each such ``$R$-vertex'' to $t$. We add one ``$f$-vertex'' for each face $f$ in $\cH^\times$ that $\lambda$ maps to $R$, and a capacity-1 edge from each such vertex to every $R$-vertex that lies on the boundary of $f$. We add an (unlimited-capacity) edge from $s$ to every $f$-vertex whose face contains $r$ (possibly on its boundary). Finally, we add an edge from every $f$-vertex whose face contains $r$ to each other $f'$-vertex of capacity equal to the number of crossable edges that lie on the shared boundary of $f$ and \(f'\).
	The instance $\theta_2$ is constructed in an analogous fashion for $B$ and $b$.
		\ifshort{
		To conclude the proof, it suffices to show that the drawings of the edges in \(\aE\) incident to \(r\) in a \(\lambda\)-consistent extension \(\cG\) of \(\cH\) to \(G\) correspond to \(s\)-\(t\)-flows of values \(|N_{E(G) \setminus E(H)}(r)|\) and \(|N_{E(G) \setminus E(H)}(b)|\) for \(\theta_1\) and \(\theta_2\), respectively.
	}
	\iflong{
	
	Now, assume there is a $\lambda$-consistent extension \(\cG\) of \(\cH\) to \(G\). Each edge from $r$ to an endpoint $v\in R$ is either non-crossing (in which case $\theta_1$ models it as a direct flow from $s$ to $v$ through a face containing $r$, and then to $t$), or crosses into another face $f_2$ (in which case $\theta_1$ models it as a flow from $s$ to $f_1$, then to $f_2$, then to $v$, and finally to $t$). The fact that each such edge in $\cG$ must use a separate crossable edge when crossing out of a face containing $r$ ensures that routing the flow in this way will not exceed the edge capacities of $\theta_1$. The argument for $\theta_2$ is analogous, and hence we obtain that both $\theta_1$ and $\theta_2$ allow \(s\)-\(t\)-flows of values \(|N_{E(G) \setminus E(H)}(r)|\) and \(|N_{E(G) \setminus E(H)}(b)|\) respectively.
	
	Conversely, assume that both $\theta_1$ and $\theta_2$ \(s\)-\(t\)-flows of values \(|N_{E(G) \setminus E(H)}(r)|\) and \(|N_{E(G) \setminus E(H)}(b)|\) respectively. Consider a witnessing flow for $\theta_1$ (the procedure for $\theta_2$ will be analogous). Clearly, this flow must route 1 capacity through each $R$-vertex $v$ to achieve its capacity bound. If the flow enters $v$ after going only through a single $f$-vertex, then add a drawing of the edge from $r$ to $v$ passing only through $f$. If it instead enters $v$ after passing through $f$-vertices (say $f_1$ and $f_2$), then add a drawing of the edge from $r$ to $v$ beginning in $f_1$ and crossing into $f_2$. Since no edge from $b$ enters a face marked $R$, there is a way of choosing which boundaries to $f_2$ to cross in order to ensure that edges from $r$ will not cross each other (the order of crossings matches the order on which the target vertices $v\in R$ appear on the boundary of $f_2$). Applying the same argument for $\theta_2$ results in a $\lambda$-consistent extension $\cG$ of \(\cH\) to \(H\).
	}
\end{proof}

\begin{corollary}
	\textsc{1-Planar Drawing Extension} for \(|\aV| = 1\) and \(\aEpar = \emptyset\) can be solved in polynomial time.
\end{corollary}
\iflong{
\begin{proof}
	Branch on the cell containing the single added vertex \(r\) in \(\bigoh(n)\)-time.
	Then apply Lemma~\ref{lem:flow} to \(\cH\), \(G\) and \(\lambda\) where \(\lambda\) maps every face of \(\cH^\times\) to \(R\).
\end{proof}
}

\ifshort{\subparagraph{Finding an Initial Delimiter.} We begin by formally defining the following notion:}
\iflong{
\subsection{Finding an Initial Delimiter}
If $\aE$ contains the edge $rb$, we then branch to determine which edge it crosses (if any), and add $rb$ into $\cH$ as well. }

\begin{definition}
A Jordan curve $\omega$ in the plane is an \emph{initial delimiter} for $\cH$ if:
\begin{enumerate}
\item $\omega$ passes through both $r$ and $b$ but through no other vertex of $\cH$, 
\item whenever $\omega$ shares at most one point with the interior of an edge, this point is a proper crossing between $\omega$ and that edge,
\item the intersection between $\cH$ and the exterior of $\omega$ (including $\omega$ itself) contains a single cell $c_r$ whose boundary contains $r$, and a single cell $c_b$ whose boundary contains $b$, and
\item the intersection of the boundary of $c_r$ (resp. $c_b$) and $\omega$ is a single simple curve containing $r$ (resp. $b$) as an interior point. 
\end{enumerate}
\end{definition}

Intuitively, the third condition means that if we add $\omega$ onto $\cH$, then there are unique cells in the exterior of $\omega$ for $r$ and $b$. A solution for our instance, i.e., a drawing $\cG$ of $G$, is $\omega$-\emph{compatible} if every edge from $\aE$ is drawn in the exterior region defined by $\omega$.\ifshort{
We state the main result of this subsection below---intuitively, it provides us with a set of initial delimiters that we can exhaustively branch over, and in each branch we can restrict our attention to solutions that are compatible with the chosen initial delimiter. To avoid confusion, we note that the lemma covers the case where $r$ and $b$ are placed in the same cell $f$ (by branching on and placing a constant number of new edges that ``separate'' the boundary of $f$).
%
%
%
}
\iflong{

Our aim will be to obtain a set $Q$ of initial delimiters for which we know that if the instance $(G,H,\cH)$ admits a solution, then $Q$ contains at least one $\omega$ such that $(G,H,\cH)$ also admits an $\omega$-compatible solution. This will later allow us to exhaustively branch over all delimiters in $Q$, and in each branch restrict our attention to merely finding an $\omega$-compatible solution.

If $r$ and $b$ were placed in different cells and $rb$ was part of $\aE$ (i.e., is part of $\cH$ now), then we can construct a single initial delimiter with the desired property. Indeed, due to previous branching we have already added the edge $rb$ into $\cH$, and since this edge is crossing we know that no other edge in $\aE$ can cross $rb$. Hence, we can obtain the Jordan curve $\omega$ by starting at $r$, following the edge $rb$ (without touching it, staying close to the ``right'') to $b$ and then following $rb$ back to $r$ on the other side. This way, the interior of $\omega$ intersects $\cH$ in precisely drawing of the edge $rb$ and a small part of the single edge that $rb$ crosses. Since in every potential solution, no other edge in $\aE$ intersects this part of $\cH$, we can restrict our attention only to $\omega$-compatible solutions.
 Moreover, the single edge that $\omega$ splits into two cannot be crossed by any edge in a solution (this observation will be useful later). 

The next case we will consider is that $r$ and $b$ are placed in the same cell \(c\).
In this case there may be up to two pairs of edges incident to \(r\) and \(b\) whose drawings cross each other in a solution.
Branching in \(\bigoh(n^4)\) on the existence and identity of these edges completely determines their drawings, which we from now on assume to be part \(\cH\) and adapt \(H\) correspondingly.
After this, in a hypothetical 1-plane extension of \(\cH\) to \(G\) we can identify four edges \(e_r, f_r, e_b, f_b \in E(G) \setminus E(H)\) such that
	\begin{itemize}
		\item \(e_r\) and \(f_r\) are consecutive edges in \(E(G) \setminus E(H)\) in the rotation system around \(r\);
		\item \(b\) lies in or on the boundary of the part \(c'_b\) of \(c\) which lies between \(e_r\) and \(f_r\) in clockwise direction;
	\end{itemize}
	and analogously
	\begin{itemize}
		\item \(e_b\) and \(f_b\) are consecutive edges in \(E(G) \setminus E(H)\) in the rotation system around \(b\); and
		\item \(r\) lies in or on the boundary of the part \(c'_r\) of \(c\) which lies between \(e_b\) and \(f_b\) in clockwise direction.
	\end{itemize}
	In this situation the outer boundary of \(c'_r \cap c'_b\) is a valid choice for \(\omega\):
	Drawings of edges in \(E(G) \setminus E(H)\) can not leave and reenter \(c\) due to 1-planarity.
	Thus drawings of edges incident to \(r\) and \(b\) in \(\cG\) intersect the interior of \(c\) if and only if corresponding edges incident to \(r\) and \(b\) in \(\pG\) intersect the interior of \(c\).
	Since \(e_r\) and \(f_r\) are consecutive, there are no edges in \(E(G^\times) \setminus E(H^\times)\) incident to \(r\) which are drawn in \(c'_b\) by \(\pG\).
	Similarly no edges in \(E(G^\times) \setminus E(H^\times)\) incident to \(b\) are drawn in \(c'_r\).
	The choice and drawing of \(e_r, f_r, e_b\) and \(f_b\) can be branched on in \(\bigoh(n^8)\) ways. 

Finally, consider the case where $r$ and $b$ were placed in different cells and $rb$ is not an edge in $\cH$. This is the most difficult case for us, and the rest of this subsection is devoted to resolving it. Let $\bar \cH$ be the graph obtained from the planarization $\cH^\times$ of $\cH$ by creating a \emph{face-vertex} for each face of $\cH^\times$, an \emph{edge-vertex} for each edge in $\cH^\times$, and making each face-vertex representing a face $f$ adjacent to all edge-vertices that represent edges which lie on the boundary of $f$.

An $F_r$-$F_b$ path $P$ in $\bar \cH$ \emph{corresponds} to an initial delimiter $\omega$ if $\omega$ is drawn by starting from $r$, following $P$ (with edge-vertices specifying where $\omega$ should cross, and face-vertices specifying faces $\omega$ enters into), reaching $b$, and from there backtracking along $P$ to $r$ (in a similar way as we followed the $rb$ edge in the case it was part of $\cG$). Clearly, an initial delimiter that corresponds to a given $F_r$-$F_b$ path $P$ can be constructed in polynomial time.

One observation we will use is that the ``exact drawing'' of an initial delimiter corresponding from an $r$-$b$ path does not matter. In particular, if $\cG$ is a solution and $P$ is an $r$-$b$ path with some corresponding initial delimiter $\omega$ such that $\cG$ is $\omega$-compatible, then for \emph{every} initial delimiter $\omega'$ corresponding to $P$, we can topologically shift the drawing of the edges in $\aE$ to obtain an $\omega'$-compatible solution.

Now, consider the subcase where a shortest path in $\bar \cH$ between the face $F_r$ of $\cH$ containing $r$ and the face $F_b$ of $\cH$ containing $b$ consists of at least $6$ edges, i.e., contains at least two faces $f_1$, $f_2$ distinct from $F_r$ and $F_b$. Then, each face of $\cH$ shares either a boundary of $F_r$ or $F_b$, but not both. In this case, we can solve the instance via an application of Lemma~\ref{lem:flow} to \(\cH\), \(G\) and \(\lambda\), where \(\lambda\) maps the faces of \(\cH^\times\) at distance at most \(1\) from \(F_r\) in \(\bar{\cH}\) to \(R\), and the faces of \(\cH^\times\) at distance at most \(1\) from \(F_b\) in \(\bar{\cH}\) to \(B\).

On the other hand, we reach a situation where we are guaranteed that there exists a path in $\bar \cH$ with at most $4$ edges. In this case, we prove:

\begin{lemma}
\label{lem:shortcurve}
Assume the instance $(G,H,\cH)$ admits a solution $\cG$ and $\bar \cH$ contains an $F_r$-$F_b$ path of length at most $4$. Then there is a solution $\cG'$ and a simple curve $C$ from $r$ to $b$ in $\cH$ that does not cross any edge in $\aE$ and crosses at most $2$ edges of $\cH$.
\end{lemma}

\begin{proof}
Let us enumerate all edges in $\aE$ as $e_1,\dots,e_\ell$. We will proceed by induction, where our induction hypothesis for $i$ is that there is a simple $r$-$b$ curve $C_i$ that does not cross any of the $e_j$ where $j\leq i$ and crosses at most $2$ edges of $\cH$. Clearly, the hypothesis holds for $i=0$ (as witnessed by any initial delimiter corresponding to the assumed $F_r$-$F_b$ path in $\bar \cH$). Moreover, if the hypothesis holds for $i=\ell$, then the lemma also holds.

So, assume that for $i\leq \ell$, such a simple curve $C_i$ exists. If $e_{i+1}$ does not cross $C_i$ in $\cG$, then we can set $C_{i+1}=C_i$ and proceed. So, assume that $e_{i+1}$ crosses $C_i$ and that, w.l.o.g., $e_{i+1}$ is incident to $r$. In that case, consider a traversal of $C_i$ from $b$ until the first time $C_i$ touches $e_{i+1}$; we will call this point $p$, and we let $C_{i+1}$ be the same as $C_i$ from $b$ up to $p$. 

We now distinguish two subcases. If $p$ occurs in a face different from $F_r$, then we let $C_{i+1}$ diverge from $C_i$ at this point and instead follow $e_{i+1}$ towards $r$. This guarantees that $C_{i+1}$ will only cross precisely one other edge in $\cH$; however, $C_i$ also had to cross at least one edge in $\cH$ from $p$ to $r$, and hence the number of edges crossed by $C_{i+1}$ is not higher than the number of edges crossed by $C_i$.

On the other hand, if $p$ occurs in $F_r$ and $e_{i+1}$ contains no crossing on its part from $r$ to $p$, then $C_{i+1}$ can simply follow $e_{i+1}$ towards $r$. The perhaps most interesting case is when $e_{i+1}$ does contain a crossing between $r$ and $p$, in spite of $p$ being in $F_r$. In this case, we will set $C_{i+1}=C_i$ but topologically shift $e_{i+1}$ without affecting any other part of the solution. In particular, we redraw $e_{i+1}$ to follow along $C_{i}$ up to $p$, and from there on let $e_{i+1}$ continue as in the original solution $\cG$. Notice that the new drawing of $e_{i+1}$ will not cross any of the edges $e_1,\dots,e_{i}$, and furthermore cannot have more crossings than $e_{i+1}$ had in $\cG$---indeed, there could be at most $1$ crossing on $C_i$ between $p$ and $r$\footnote{We may, in fact, also inductively assume that no such crossing occurs on $C_i$ in this case, but it is not necessary.}, and $e_{i+1}$ already had one crossing between $r$ and $p$ in $\cG$.
\end{proof}

Finally, observe that the ``exact drawing'' of an initial delimiter corresponding from an $F_r$-$F_b$ path does not matter. In particular, if $\cG$ is a solution and the $P$ is an $F_r$-$F_b$ path with some corresponding initial delimiter $\omega$ such that $\cG$ is $\omega$-compatible, then for \emph{every} initial delimiter $\omega'$ corresponding to $P$, we can topologically shift the drawing of the edges in $\aE$ to obtain an $\omega'$-compatible solution.

With this in hand, we can now construct the set $Q$ for this case as follows: we branch over all $F_r$-$F_b$ paths $P$ of length at most $4$ in $\bar \cH$, for each such path construct a corresponding initial delimiter $\omega$, and add $\omega$ to $Q$. We note that $\omega$ may split at most $2$ edges into two parts that can be still crossed by an edge in $\aE$. It will be clear in the next section that once we choose $\omega\in Q$, we can also branch on which half of the split edge is crossable. 
We summarize the result for this section below:
}

\begin{lemma}\label{lem:Qs}
For every instance $(G,H,\cH)$ where $|\aV|\leq 2$ and $\aEpar =\emptyset$, we can in polynomial time either solve $(G,H,\cH)$ or construct a set $Q$ of initial delimiters with the following property (or both): if $(G,H,\cH)$ admits a solution, then $Q$ contains at least one $\omega$ such that $(G,H,\cH)$ also admits an $\omega$-compatible solution.
\end{lemma}
 
\ifshort{\subparagraph{Dynamic Programming}}
\iflong{	\subsection{Dynamic Programming}}
	\label{subsubsec:dynprog}

\iflong{	
	From the previous step we have a 1-planar extension \(\cH'\) of \(\cH\) into which we want to insert \(E(G) \setminus E(H')\) such that the resulting drawing is 1-planar and drawings of inserted edges do not mutually intersect.
	Additionally we distinguish a cell \(c_r\) and a cell \(c_b\) of \(\cH'\) which we assume to contain the drawings of all edges incident to \(r\) and \(b\) respectively in planarized hypothetical solution.
	Moreover we are given an initial delimiter \(\omega\) which is a Jordan curve containing \(r\) and \(b\), which we are guaranteed will not be crossed by a drawing of an~edge inserted in any hypothetical solution.
	
	Our dynamic program will consider the planarization \(\cH'^\times\) of \(\cH'\).
	To distinguish between edges that we no longer or that we still can cross, we mark edges in \(H'^\times\) as \emph{uncrossable} or \emph{crossable} depending on whether they were subdivided in the process of planarization, i.e.\ they were crossed in \(\cH'\), or not.
	
	Formally the problem we aim to solve can be expressed in a self-contained way as:
	
	
	\begin{mdframed}\label{prob:tunnel}
		\textsc{Deliminated 2-Vertex Routing}\\
		{\itshape Instance:} \(\III = (\tilde{\cG}, r, b, F_r, F_b, R, B, \text{crossable}, \omega)\); where \(\tilde{\cG}\) is a plane graph; \(r, b \in V(\tilde{G})\); \(F_r\) is a face of \(\tilde{\cG}\) containing \(r\), and \(F_b\) is a face of \(\tilde{\cG}\) containing \(b\); \(R, B \subseteq V(\tilde{G}) \setminus \{r,b\}\); \(\text{crossable} : E(\tilde{G}) \to \{\top, \bot\}\) and \(\omega\) is a Jordan curve in \(\mathbb{R}^2\) with \(r, b \in \omega\).\\
		{\itshape Task:}  Find a 1-planar extension of $\tilde{\cG}$ which adds the edges from $r$ to each vertex in $R$ and from $b$ to each vertex in $B$ in which the drawings of these edges do not cross (1) each other, (2) $\omega$, and (3) \(\tilde{\cG}(e)\) for which \(\text{crossable}(e) = \bot\); or, correctly determine no such extension exists. 
	\end{mdframed}
	
	\smallskip
	\noindent \textbf{Instantiation.} \quad
	We will consider \textsc{Deliminated 2-Vertex Routing} for the instance \((\cH'^\times, r, b, F_r, F_b, N_{(V(G), E(G) \setminus E(H'))}(r), N_{(V(G), E(G) \setminus E(H'))}(b), \text{crossable}, \omega)\), where
	\(F_r\) and \(F_b\) are the faces of \(\cH'^\times\) corresponding to \(c_r\) and \(c_b\) respectively, and \(\text{crossable}(e) = \top\) if and only if \(e\) is crossable.
	\smallskip
	
%
	
	Note that since $r,b \in \omega$, \(\omega\) can be partitioned into two simple curves between $r$ and $b$; one of which (chosen arbitrarily) we call $\omega_\emph{start}$ and the other $\omega_\emph{end}$. Moreover, the boundary of $F_R$ (resp. $F_B$) intersect both $\omega_\emph{start}$ and $\omega_\emph{end}$ in a single simple curve with one endpoint $r$ (resp. $b$).   
	
	Let us now traverse the boundary of $F_R$ that is disjoint from $\omega$ from $\omega_\emph{start}$ to $\omega_\emph{end}$, and mark the vertices we visit on this walk via the increasing sequence $r_1,r_2,\dots r_\emph{end}$. (The vertices $b_1,b_2,\dots b_\emph{end}$ will be defined analogously, but for $b$.) Observe that it may happen that a vertex is marked as $r_i$, $r_j$ and $b_q$, for three possibly distinct integers $i,j,q$.
	
	\smallskip
	
	We note that $\III$ already contains some initial information about which edges may in fact enter which face. We say that a face of $\III$ is \emph{red} if it shares at least one crossable edge with $F_R$ but none with $F_B$, \emph{blue} if it shares at least one crossable edge with $F_B$ but none with $F_R$, and \emph{purple} if it shares at least one crossable edge with both $F_B$ and $F_R$. Moreover, we call edges that lie on the boundary of $F_B$ as well as $F_R$ \emph{green}. One property that we will need later is that, when traversing along the boundary of $F_R$ and $F_B$ from a certain ``point'' the occurrence of green edges and purple faces is ``synchronized''---we formalize this below.
	
	\begin{observation}
		Let $r_i$, $b_{i'}$ be either the same vertex, or be the highest-index vertices on the boundary of some purple face $f_1$. 
		
		For every vertex $r_{j}=b_{j'}$ (where $j>i$ and $j'>i'$), it holds that the unique $r_i$-$r_j$ walk along the boundary of $F_R$ does not contain any vertex on the boundary of a purple face if and only if the unique $b_{i'}$-$b_{j'}$  walk along the boundary of $F_B$ does not contain any vertex on the boundary of a purple face.
		
		Similarly, if $r_j$ and $b_{j'}$ are the lowest-index vertices on the boundary of some other purple face $f_2$ (where $j>i$ and $j'>i'$), it holds that the unique $r_i$-$r_j$ walk along the boundary of $F_R$ does not contain any vertex on the boundary of any other purple face if and only if the unique $b_{i'}$-$b_{j'}$  walk along the boundary of $F_B$ does not contain any vertex on the boundary of any other purple face.
	\end{observation}
	
	
	\smallskip 
	\noindent \textbf{High-level Overview and Records.} \quad
	We can now proceed with a very high-level sketch of the algorithm $\mathbb{A}$ for \textsc{Deliminated 2-Vertex Routing}. The general idea is to perform a dynamic left-to-right sweep of $\tilde{\cG}$ (more precisely, a traversal that proceeds from $r_1$ and in parallel from $b_1$ towards $\omega_\emph{end}$). The ``steps'' we use to move from one position to the next will rely on a situational combination of exhaustive branching and the network-flow subroutine described in Lemma~\ref{lem:flow}. 
	
	The intuitive reason we need to combine both of these techniques is that in some parts of our sweep, we will encounter faces where edges from both $r$ and $b$ may enter---there the interactions between these edges are too complicated to be modeled as a simple flow problem, but (as we will show) we can identify separating curves that cut our instance into easier-to-deal parts for which we have a mapping $\lambda$ that can be applied in conjunction with Lemma~\ref{lem:flow}. 
	
	We now formalize the records used by $\mathbb{A}$: a \emph{record} is a tuple $(\alpha_r,\alpha_b,T)$,
	where $\alpha_r$ is either a vertex in $G$ or an edge $e$ between $r$ and $v\in R$ that is added to $\tilde{\cG}$ and crosses a crossable edge in $\tilde{\cG}$ at most once. $\alpha_b$ is defined symmetrically. Note that we can combinatorially represent $e$ by specifying its crossing point and endpoint, since two drawings of edges with the same crossing point and endpoint have the same topological properties and hence are equivalent for the purposes of this section. $T$ is then an auxiliary element that simplifies the description of the algorithm and will be assigned one out of $5$ values (this will be detailed in the next subsection). There are also two special records, \textsc{Start} and \textsc{End}. As a consequence, it follows that the total number of records is upper-bounded by $\bigoh(|\III|^4)$.
	
	$\mathbb{A}$ will use two data structures: the set $\textsc{Reach}$ of records of $\III$ that are ``reachable'', and the set $\textsc{Proc}$ of all records  of $\III$ that have already been exhaustively ``processed''. At the beginning, we set $\textsc{Proc}=\emptyset$ and $\textsc{Reach}=\{\textsc{Start}\}$. If, at any stage, $\mathbb{A}$ reaches the situation where $\textsc{Proc}=\textsc{Reach}$ then we let $\mathbb{A}$ output ``NO''. On the other hand, if $\mathbb{A}$ adds $\textsc{End}$ to \textsc{Reach}, then it ascertains that $\III$ is a YES-instance (and outputs a solution that can be computed via standard backtracking along the successful run of the dynamic algorithm).
	
	\smallskip 
	\noindent \textbf{Record Types and Delimiters.} \quad
	The type of a record (stored in $T$) intuitively represents which out of $5$ ``cases'' the record encodes. Here, we provide a formal description of each such case. We note that for each case, it will be easy to verify whether it is compatible with a certain choice of $\alpha_r$ and $\alpha_b$---if this is not the case, we will never consider such a ``malformed'' combination $(\alpha_r,\alpha_b,T)$ in our branching and for our records.
	
	In our description, we will also define a \emph{delimiter} $D$ for each record type---this is a simple curve from $r$ to $b$ that separates the instance into two subinstances and behaves similarly as the two $r$-$b$ curves that make up $\omega$. Intuitively, in a ``correct'' branch we will later be able to assume that 
	\begin{enumerate}
		\item the delimiter splits a hypothetical solution into edges that are drawn only between $\omega_\emph{start}$ and $D$ and edges that are drawn only between $D$ and $\omega_\emph{end}$, and
		\item we have an assignment $\lambda$ that maps all faces between $D$ and the ``previous'' delimiter $D$ (or $\omega_\emph{start}$, at the very start of the algorithm) to $\{R,B\}$.
	\end{enumerate}
	
	We now describe the five types of our records. Note that all descriptions are given from the ``perspective'' of $r$, and symmetric cases also exist for $b$.
	
	\begin{figure}
		\begin{subfigure}{.31\textwidth}
			\includegraphics[page=2]{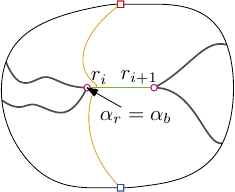}
			\caption{Green pointer.}
			\label{fig:records:case1}
		\end{subfigure}
		\hfill
		\begin{subfigure}{.31\textwidth}
			\includegraphics[page=3]{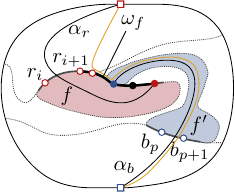}
			\caption{Double incursion.}
			\label{fig:records:case2}
		\end{subfigure}	
		\hfill
		\begin{subfigure}{.31\textwidth}
			\includegraphics[page=4]{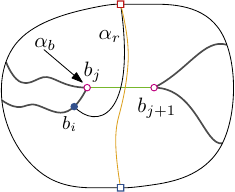}
			\caption{Left incursion.}
			\label{fig:records:case3}
		\end{subfigure}
		
		\begin{subfigure}{.31\textwidth}
			\includegraphics[page=5]{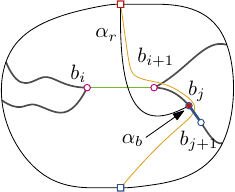}
			\caption{Right incursion.}
			\label{fig:records:case4}
		\end{subfigure}	
		\begin{subfigure}{.31\textwidth}
			\includegraphics[page=6]{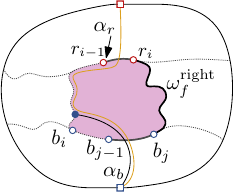}
			\caption{Slice.}
			\label{fig:records:case5}
		\end{subfigure}
		
		\caption{The five types of records for the dynamic program. The orange curves are the delimiters computed in each case. $ r $ and $ b $ are the rectangular vertices, colored disks are the sets $ B $ and $ R $, while white disks are the vertices $ r_i $ and $ b_i $ depending on their border-color.}
		\label{fig:records}
	\end{figure}
	
	\smallskip
	\noindent \textbf{1. Green Pointer.} $\alpha_r=\alpha_b$ are the same vertex $r_i$, and that vertex is incident to a green edge $r_ir_{i+1}$.
	
	\noindent \emph{Delimiter:} The delimiter in this case is a simple curve that starts in $r$, crosses through $r_ir_{i+1}$, and then proceeds to $b$.
	
	\noindent \emph{Intuition:} This is the simplest type of our records, and is only used in a few boundary cases to signify that ``nothing noteworthy happened until here''. When using this record as well as any other record, we will also be following the assumption that no edge of the solution crosses the delimiter.
	
	\smallskip
	\noindent \textbf{2. Double Incursion.} $\alpha_r$ is an edge which crosses an edge $r_ir_{i+1}$ into some face $f$ (which may, but need not, be $F_B$); let $\omega_f$ be the walk around the boundary of $f$ in the direction $r_i\rightarrow r_{i+1}$ that ends at the endpoint of $\alpha_r$. $\alpha_b$ is an edge which ends on $\omega_f$ and crosses into a face $f'$ (which may, but need not, be $F_R$) through some edge $b_pb_{p+1}$. 
	
	\noindent \emph{Delimiter:} The delimiter in this case is a simple curve that starts in $b$, crosses through $b_pb_{p+1}$ (intuitively, this occurs ``behind'' the crossing point of $\alpha_b$), and then follows $\alpha_b$ (without crossing it) until reaching the endpoint of $\alpha_b$. At that point, it follows $\omega_f$ in reverse order (without crossing it) towards $r_i$. At some point, the delimiter must reach an edge that is on the boundary of $F_R$ (since $r_ir_{i+1}$ is such an edge; that being said, the delimiter might already be in $F_R$ when it reaches the endpoint of $\alpha_b$). When that happens, it diverges from the boundary and heads to $r$. Notice that even if the delimiter crosses through several purple faces, only at most one of these will remain purple in the area enveloped by the delimiter and $\omega_\emph{end}$.
	
	\noindent \emph{Intuition:} This is the most complicated record, since it covers a number of cases of non-trivial interactions between $\alpha_b$ and $\alpha_r$. When using it, we will be assuming that our edges have a maximality property which ensures that no edge from $b$ ends between $r_{i+1}$ and the endpoint of $\alpha_b$, and no edge from $r$ ends between the endpoint of $\alpha_r$ and $b_pb_{p+1}$. The edge $\alpha_b$ will split whatever face it crosses into one part that only be accessible to edge from $b$, and one part that will only be accessible to edges from $r$. $\alpha_r$ induces a similar splitting of the face it crosses into.

	%
	%
	
	\smallskip
	\noindent \textbf{3. Left Incursion.} $\alpha_r$ is an edge $e$ which ends in $b_i$ after crossing across a green edge $b_{j}b_{j+1}$, where $j>i$. $\alpha_b=b_{j}$.
	
	\noindent \emph{Delimiter:} The delimiter in this case is a simple curve that starts in $r$, crosses through $b_jb_{j+1}$ (intuitively, this occurs ``behind'' the crossing point of $e$), and then proceeds to $b$.
	
	\noindent \emph{Intuition:} This record represents the case where only one ``half'' of a \textbf{Double Incursion} is present. $e$ splits $F_B$ into a pseudo-red and a pseudo-blue face, while $F_R$ will remain only accessible to $r$ to the left of the delimiter.
	
	\smallskip
	\noindent \textbf{4. Right Incursion.} $\alpha_r$ is an edge $e$ which crosses an edge $b_ib_{i+1}$ and ends in $b_j$ (where $j>i$) (and the previous cases do not apply). $\alpha_b=b_j$.
	
	\noindent \emph{Delimiter:} The delimiter in this case is a simple curve that starts in $b$, crosses through $b_jb_{j+1}$ (intuitively, this occurs ``behind'' $b_j$), and then proceeds along the boundary of $F_B$ (without crossing it) towards $b_i$. Since $b_ib_{i+1}$ is a green edge, it must eventually cross into $F_R$ (it might, in fact, already be there); when that happens, it diverges from the boundary and heads to $r$.
	Notice that even if the delimiter crosses through several purple faces, only at most one of these will remain purple in the area enveloped by the delimiter and $\omega_\emph{end}$.
	
	\noindent \emph{Intuition:} This record represents the case where only one ``half'' of a \textbf{Double Incursion} is present, and can be viewed as a degenerate \textbf{Double Incursion} where $\alpha_b$ is collapsed into a single vertex.
	
	\smallskip
	\noindent \textbf{5. Slice.} $\alpha_r$ is either a vertex on the boundary of a purple face $f$, or an edge which crosses into $f$. $\alpha_b$ is either a vertex on the boundary of $f$ or an edge which crosses into $f$. Moreover, if $\omega^\emph{right}_f$ is the unique walk from the maximum-index $r_i$ (with a crossable edge to $F_R$) on the boundary of $f$ to the maximum-index $b_j$ (with a crossable edge to $F_B$) on the boundary of $f$ that avoids the lowest-index $b_i$ on the boundary of $f$, then neither $\alpha_b$ nor $\alpha_r$ end on $\omega^\emph{right}_f$.
	
	\noindent \emph{Delimiter:} Let $\omega_f$ be the extension of $\omega^\emph{right}_f$ to a closed walk around the whole boundary of $f$, starting from and ending at $r_i$ (note that the boundary of $f$ might include a ``previous delimiter'', which we then treat as part of the boundary). The delimiter starts in $b$. If $\alpha_b$ is a vertex $b_p$, it crosses through $b_{p}b_{p+1}$ and then continues along $\omega_f$. If $\alpha_b$ is an edge, let \texttt{b} be the first time an endpoint or crossing point of $\alpha_b$ is found when walking along $\omega_f$; the delimiter then crosses just to the right of \texttt{b}, continues by following $\alpha_b$ into $f$ (without crossing it), and once it returns to the boundary of $f$ it then continues along $\omega_f$.
	Once the delimiter reaches an endpoint or crossing point of $\alpha_r$ in case \(\alpha_r\) is an edge (whichever comes first); or \(\alpha_r\) itself otherwise, it proceeds analogously as for $\alpha_b$ to reach $r$.
	
%
	\noindent \emph{Intuition:} This record represents, in a unified way, a multitude of situations that may occur inside a purple face. It will be used in a context where $\alpha_r$ and $\alpha_b$ satisfy a certain maximality condition which ensures that, in the part up to the delimiter (i.e., between the previous and current delimiter), edges from $b$ only enter $f$ ``under'' $\alpha_b$, and similarly edges from $r$ only enter $f$ ``under'' $\alpha_r$.
	
	\smallskip
	\noindent \textbf{Branching Steps.}
	We can now complete the description of $\mathbb{A}$. The top-level procedure iterated by $\mathbb{A}$ is that it picks some record from $\textsc{Reach}\setminus\textsc{Proc}$ (at the beginning this is \textsc{Start}, but later it will be some $(\alpha'_r,\alpha'_b,T)$). It then computes a delimiter $D'$ for that record (in the case of \textsc{Start}, $D'$ is $\omega_\emph{start}$).
	
	Now, $ \mathbb{A} $ exhaustively branches over all possible choices of $\alpha_r$, $\alpha_b$, and $T$, and checks whether the record $(\alpha_r, \alpha_b,T)$ satisfies the conditions for the given type $T$ (i.e., whether it is not malformed). It then adds $\alpha_r$, $\alpha_b$ into the drawing, computes the delimiter $D$ for $(\alpha_r, \alpha_b,T)$, and checks that $D$ occurs ``to the right'' of $D'$
	(i.e., that $D'\setminus\{r, b\}$ is completely contained in the area delimited by $\omega_\emph{start}$ and $D$)---otherwise, it discards this choice. It then updates the colors of the faces that lie between $D'$ and $D$ (e.g., a purple face in $\III$ might become a blue face due to no longer having a crossable edge to $F_R$ between $D'$ and $D$).
	
	Finally, $\mathbb{A}$ checks that there is no ``untouched'' purple face between $D'$ and $D$: in particular, it checks that every purple face encountered on the walk along the boundary of the face that now contains $r$ intersects with either $D$ or $D'$; if this is not the case, it discards this branch.
	
	If all of the above checks were successful, $\mathbb{A}$ constructs a mapping $\lambda$ to verify whether the branch represents a valid step of the algorithm---notably, whether the edges that must lie between $D'$ and $D$ may be routed in line with the consistency criteria corresponding to the informal intuitions provided in the description of the $5$ cases. This is done as follows (assuming the cases are given w.r.t.\ $r$, as in the description of the cases; mirrored cases are treated analogously, but with swapped $F_B$ and $F_R$):
	
	\begin{enumerate}
		\item If $T=\textbf{Green Pointer}$, set $\lambda(F_R)=R$ and $\lambda(F_B)=B$;
		\item If $T=\textbf{Double Incursion}$, $\alpha_r$ splits the face $f$ it crosses into into two parts. Let $r_ir_{i+1}$ be the edge $\alpha_r$ crossed. Then the unique walk along the boundary of $f$ from $r_i$ to the endpoint of $\alpha_r$ that contains $r_{i+1}$ occurs in one part of $f$ that is separated from the other by $\alpha_r$; this part is mapped to $R$, while the other is mapped to $B$. $\alpha_b$ splits the face $f'$ it crosses into into two parts, then the part of $f'$ that lies inside the delimiter (and hence lies ``below'' $\alpha_b$) is mapped to $B$.
		\item If $T=\textbf{Left Incursion}$, $\alpha_r$ splits $F_B$ into two parts, where the part containing $b$ is mapped to $B$ and the other part to $R$. $F_R$ is mapped to $R$.
		\item If $T=\textbf{Right Incursion}$, $\alpha_r$ splits $F_B$ into two parts, where the part containing $b$ is mapped to $B$ and the other part to $R$. $F_R$ is mapped to $R$. While the boundary may pass through a purple face $f$, we will leave it unmapped (i.e., no edge can cross through it).
		\item If $T=\textbf{Slice}$, $\alpha_r$ splits the purple face $f$ it crosses into into two parts, and we map the (uniquely defined) part of $f$ that does not contain $\omega^\emph{right}_f$ to $R$. Similarly, we map the part of $f$ delimited by $\alpha_b$ that does not contain $\omega^\emph{right}_f$ to $B$. We leave the rest of $f$ unmapped, and map $F_B$ and $F_R$ to $B$ and $R$, respectively. In the case that $\alpha_r$ or $\alpha_b$ are vertices, we do not map anything to $R$ or $B$, respectively.
	\end{enumerate}
	
	In addition to the above, we naturally let $\lambda$ map red faces to $R$ and blue faces to $B$. At this point, the only purple face that still remains unresolved by $\lambda$ is the potential purple face crossed by $D'$, in the case where it is not crossed by $D$. If such a face $f'$ exists, we perform an additional branching subroutine: we loop over each potential edge $\beta_r$ from $r$ that cross into $f'$ and ends in $\omega^\emph{right}_{f'}$ (with a special case of setting $\beta_r=\emptyset$, in case no such edge exists), and the same for $\beta_b$ (from $b$). We then use these to split $f'$ into three parts, and have $\lambda$ map these as follows: the unique part delimited by $\beta_r$ that contains the maximum-index $r_i$ in $f'$ is mapped to $R$, the unique part delimited by $\beta_b$ that contains the maximum-index $b_j$ in $f'$ is mapped to $B$, and the rest of $f'$ is unmapped.
	
	Finally, with this mapping, we invoke Lemma~\ref{lem:flow} on the (sub-)instance delimited by $D$ and $D'$ after adding $\alpha_r,\alpha_b$ and potentially also $\beta_r,\beta_b$ to $\tilde{\cG}$. If the resulting network flow instance is successful, we add $(\alpha_r,\alpha_b,T)$ to \textsc{Reach}. Once all branches have been exhausted for a given record, we add it to \textsc{Proc}.
	
	We also include the special record \textsc{End} among the records we branch over, where the delimiter is simply $\omega_\emph{end}$. If $\mathbb{A}$ adds \textsc{End} to \textsc{Reach}, then \textsc{A} outputs ``YES'' (and backtracks to find a suitable drawing by inductively invoking Lemma~\ref{lem:flow} on the relevant (sub-)instances of $\III$), and if it ends up in a situation where \textsc{Proc}=\textsc{Reach} then $\mathbb{A}$ rejects $\III$. The total runtime of $\mathbb{A}$ is upper bounded by $\bigoh(|\III|^{12})$, since the number of records is at most $\bigoh(|\III|^4)$ and the branching subroutines take time at most $\bigoh(|\III|^8)$.
	
	\smallskip
	\noindent \textbf{Domination and Correctness.}
	All that remains now is to show that $\mathbb{A}$ is, in fact, correct. To do so, we will need the notion of \emph{undominated edges}, which formalizes the maximality assumption that we informally hinted at in the intuitive description of our record types.
	Given a solution $\Phi$ to $\III$, an edge $e$ starting in $r$ is \emph{dominated} if the following holds:
	\begin{itemize}
		\item $e$ crosses into $F_B$, and the face it encloses with the boundary of $F_B$ that does not contain $b$ is contained in the face enclosed by another edge from $r$ and the boundary of $F_B$ not containing $b$;
		\item $e$ crosses into a purple face $f$, ends on the part of $\omega_f$ between the maximum-index $r_i$ and the minimum-index $b_j$ on the boundary of $f$, and the face it encloses with the boundary of $f$ that does not contain $\omega_f^\emph{right}$ is contained in the analogously defined face of another edge with the same properties as $e$.
	\end{itemize}
	
	An edge is \emph{undominated} if it is not dominated. We extend the notion of domination towards vertices by considering a vertex to be a loop of arbitrarily small size (hence, every edge ``enclosing'' that vertex dominates it). The role of undominated edges (and vertices) is that they will help us identify which parts of a hypothetical solution $\Phi$ to focus on in order to find a relevant step of $\mathbb{A}$ for our correctness proof.

	
	\begin{lemma}
	\label{lem:correct}
		$\mathbb{A}$ is correct.
	\end{lemma}

	\begin{proof}
		Assume $\mathbb{A}$ outputs ``YES''. Then one can follow the sequence of branching and records considered by $\mathbb{A}$ that leads to $\textsc{End}$ being added to \textsc{Reach}. During our branching we make sure to only consider records whose edges do not cross the delimiter of the previous record, and so an iterative application of Lemma~\ref{lem:flow} on each of our branching steps guarantees that if we follow this sequence we will end up routing all edges from $r$, $b$ to their required endpoints---i.e., with a solution to $\III$.
		
		On the other hand, assume $\III$ has a solution $\Phi$. To complete the proof, we will inductively build a sequence of records whose delimiters are not crossed by any edge in $\Phi$. This property is trivially true for the delimiter of our initial record \textsc{Start}, but in general we will simply assume we have found an arbitrary delimiter $D'$ that has this property and added it to \textsc{Reach}. It now suffices to show that $\mathbb{A}$ will be able to find (via branching) and add an additional record to its set \textsc{Reach} that pushes the delimiter further towards $\omega_\emph{end}$; since the total number of records is bounded, this would necessarily imply that $\mathbb{A}$ would terminate by adding \textsc{End} to \textsc{Reach}.
		
		We will now perform an exhaustive case distinction, where for each case we will identify a record that would be detected by branching and where the existence of $\Phi$ would guarantee that the network flow instance constructed by Lemma~\ref{lem:flow} using the corresponding $D$ and $\lambda$ (constructed as per the description of our branching steps) terminates successfully. 
		
		Let us begin by considering the case where $D'$ does not cross through any purple face. Let $r_i$ be the minimum-index vertex to the right of $D'$ which is:
		\begin{enumerate}[\text{Case} 1]
			\item the endpoint of an edge $e$ starting in $b$ that crosses into $F_R$, or
			\item the endpoint of an edge $e$ starting in $r$ that crosses into $F_B$, or
			\item incident to a green edge $e=r_{i-1}r_i$ crossed by an edge starting in $r$, or
			\item incident to a green edge $e=r_{i-1}r_i$ crossed by an edge starting in $b$, or
			\item a vertex on the boundary of a purple face $f$, or
			\item no such vertex exists.
		\end{enumerate}
		
		In Case 1 and 2, this gives rise to a \textbf{Left Incursion}. In both cases, it is easy to verify that $\Phi$ is a $\lambda$-consistent solution (and hence the branching would detect the corresponding record and add it to \textsc{Reach}, as required).
		
		In Case 3. and 4., this gives rise to a \textbf{Right Incursion} or \textbf{Double Incursion}, depending on whether there is an edge $e'$ from $b$ (or $r$, respectively) ending on the walk between the crossing point and the endpoint of $e$. If no such edge exists, we obtain a \textbf{Right Incursion}. Moreover, it is once again easy to verify that $\III$ is $\lambda$-consistent: the ``inside'' of $e$ within $F_B$ crosses into is not reachable from $b$ directly, and also not via the boundary with $F_R$ due to the non-existence of $e'$. Hence the new delimiter will not be crossed by $\Phi$ and $\Phi$ will be $\lambda$-consistent.
		
		On the other hand, if an edge $e'$ as described above exists, let us fix $e'$ to be the unique undominated edge with this property. Then the ``inside'' of $e'$ within the face it crosses into is only reachable from $b$, while the ``outside'' of $e'$ cannot be reached by $b$ anymore (due to $e'$ being undominated). The argument for $e$ is the same as above, and putting these two together we obtain that $\Phi$ is also $\lambda$-consistent and $D$ is not crossed.
		
		If Case 5. applies to $r_i$ as well as to the symmetrically constructed $b_i$, then the pair forms a (degenerate) \textbf{Slice} in $f$. Let $e$ be the unique undominated edge dominating $r_i$ (if none exists, we use $r_i$ instead), and the same for $e'$ and $b_i$. In this case, the record $(e,e',\textbf{Slice})$ would pass our branching test. Finally, in case 6. the existence of $\Phi$ guarantees that \textsc{End} would pass our branching test, and hence $\mathbb(A)$ would correctly identify a solution to $\III$.
		
		Now, consider the case where $D'$ does cross through a purple face $f$. Let $r_i$ be the minimum-index vertex to the right of $D'$ which is:
		\begin{enumerate}[\text{Case} A]
			\item the endpoint or crossing point of an edge that enters into $f$, o	r
			\item be incident to the green edge $r_ir_{i+1}$, or
			\item a vertex on the boundary of a different purple face $f'$, or
			\item no such vertex exists.
		\end{enumerate}
		
		Case C is handled analogously as Case 5 in the previous case distinction, and Cases B and D (which is itself analogous to Case 6. from before) are also simple. The by far most complicated remaining case is Case A, which has several subcases. 
		
		First, we consider what happens if both $r_ir_{i+1}$ and $b_ib_{i+1}$ are crossing points of edges $e,e'$ (respectively) which end in $\omega_f^\emph{right}$. Then, let us consider which of Case 2, Case 3, and Case 4 occurs \emph{after} $f$ (i.e., on the walks around $R_F$ and $R_B$ that proceed after leaving the last vertex on the boundary of $f$): we will use these to define our new record, say $(\alpha_r,\alpha_b,T)$. For this record, it is just as easy to argue that the delimiter is not crossed by $\Phi$ as before. As for $\lambda$, the basic mapping associated with $(\alpha_r,\alpha_b,T)$ would leave the face $f$ ``open''---that is where the additional branching rule via $\beta_r$ and $\beta_b$ comes into play. In the case where $\beta_r=e$ and $\beta_b=e'$, the obtained branch results in a mapping $\lambda$ such that $\Phi$ is $\lambda$-consistent.
		
		If the previous subcase of Case A does not occur, we distinguish the special case where an edge $e$ from $r$ crosses into $f$ via $r_ir_{i+1}$ and an edge $e'$ crosses into a face different from $f$ and ends on the walk from $r_{i+1}$ to the endpoint of $e$---in this case, the edges $e$ and $e'$ represent a \textbf{Double Incursion}. Crucially, since $e$ does not end on $\omega_f^\emph{right}$ (as that was handled by the previous subcase), $e'$ must cross into the face $F_R$. As before, we will assume that $e'$ is undominated. The resulting $\lambda$ then splits $F_R$ and $f$ in a way that is forced by $e$ (up to $D$, no red edge can exist in $f$ outside of the area delimited by $e$) and $e'$ (up to $D$, no blue edge can exist in $F_R$ outside of the area delimited by $e'$).
		
		Finally, if none of the previous subcases of Case A occur, then we are left with a situation that will boil down to a \textbf{Slice}---the last obstacle remaining is to identify the relevant edges. In particular, let $e$ be the endpoint of the edge that crosses near or ends at $r_i$, and let $e'$ be the analogously defined edge for the minimum-index $b_j$ found on $f$. If both edges have the same endpoint, then it is easy to see that they must be the same edge, and that will be the (single) edge used to define our \textbf{Slice}. If, w.l.o.g., $e$ started in $r$ and $e'$ in $b$, then we make the following check: does the part of $f$ split by $e$ not containing $\omega_f^\emph{right}$ contain the part of $f$ split by $e'$ not containing $\omega_f^\emph{right}$? If this is not the case in either direction, then the new \textbf{Slice} is defined by both edges; otherwise, it is defined by the single edge that delimits the smaller part of $f$ not containing $\omega_f^\emph{right}$. In all of these cases, it is easy to verify that $D$ indeed is not crossed by $\Phi$ and that $\Phi$ is a $\lambda$-consistent solution.
	\end{proof}
	
As a direct consequence of Lemma~\ref{lem:correct} and Lemma~\ref{lem:Qs}.
	\xpkappaonep
}

\section{Concluding Remarks}


In this paper, we initiated the study of the problem of extending partial 1-planar and IC-planar drawings by providing several parameterized algorithms that target cases where only a few edges and/or vertices are missing from the graph. Our results follow up on previous seminal work on extending planar drawings, but the techniques introduced and used here are fundamentally different~\cite{adfjkp-tppeg-15}.
The by far most prominent question left open in our work concerns the (not only parameterized, but also classical) complexity of \textsc{1-Planar Extension} w.r.t.\ $\kappa$. In particular, can one show that the problem is, at least, polynomial-time tractable for fixed values of $\kappa$? While the results presented in Section~\ref{sec:2vtcs} are a promising start in this direction, it seems that new ideas are needed to push beyond the two-vertex case.



\iflong{Another interesting direction for future research would be to improve the running times of the algorithms and establish tight lower bounds tied to the Exponential Time Hypothesis.} Follow-up work may also focus on extending\iflong{ partial $t$-planar drawings for $t\ge 2$ or} other types of beyond planar drawings~\cite{dlm-sgdbp-19}.

\bibliographystyle{plainurl}
\bibliography{references}
	
\end{document}